\documentclass[12pt]{article}
\usepackage[english]{babel}
\usepackage{amssymb,amsmath,graphicx}
\usepackage{amsthm}
\usepackage{float}
\usepackage{afterpage}
\usepackage{dcolumn}
\usepackage{multirow}
\usepackage{color}
\usepackage[numbers]{natbib}
\usepackage{enumitem}
\usepackage{slashbox}
\usepackage{pict2e}
 \usepackage{setspace}
\usepackage{siunitx} 
\usepackage[noend]{algorithmic}
\usepackage{algorithm}
\usepackage{rotating}

\newtheorem{theorem}{Theorem}
\newtheorem{lemma}{Lemma}

\newtheorem{corollary}{Corollary}

\usepackage[table]{xcolor}
\usepackage{float}
\usepackage{color} 
\usepackage{multirow}
\usepackage{graphicx}

\usepackage{flafter}

\def\sqw{\hbox{\rlap{\leavevmode\raise.3ex\hbox{$\sqcap$}}$%
\sqcup$}}
\def\sqb{\hbox{\hskip5pt\vrule width4pt height6pt depth1.5pt%
\hskip1pt}}

\def\qed{\ifmmode\hbox{\hfill\sqb}\else{\ifhmode\unskip\fi%
\nobreak\hfil
\penalty50\hskip1em\null\nobreak\hfil\sqb
\parfillskip=0pt\finalhyphendemerits=0\endgraf}\fi}
\def\cqfd{\ifmmode\sqw\else{\ifhmode\unskip\fi\nobreak\hfil
\penalty50\hskip1em\null\nobreak\hfil\sqw
\parfillskip=0pt\finalhyphendemerits=0\endgraf}\fi}

\usepackage{algorithmic}
\usepackage{algorithm}

\title{A Decomposition Algorithm for Nested Resource Allocation Problems}

\author{Thibaut Vidal, Patrick Jaillet, Nelson Maculan}

\begin{document}

\begin{center}

\begin{LARGE}
A Decomposition Algorithm for Nested Resource Allocation Problems
\end{LARGE}

\vspace*{0.85cm}

\textbf{Thibaut Vidal *} \\
Laboratory for Information and Decision Systems, Massachusetts Institute of Technology, Cambridge, MA, United States  \\
vidalt@mit.edu \\
\vspace*{0.2cm}
\textbf{Patrick Jaillet} \\
Department of Electrical Engineering and Computer Science, Laboratory for Information and Decision Systems, Operations Research Center, Massachusetts Institute of Technology, Cambridge, MA, United States  \\
jaillet@mit.edu \\
\vspace*{0.2cm}
\textbf{Nelson Maculan} \\
COPPE - Systems Engineering and Computer Science, Federal University of Rio de Janeiro,  Brazil  \\
maculan@cos.ufrj.br \\

\vspace*{0.9cm}

\begin{large}
Working Paper, MIT -- April 2014
\end{large}

\vspace*{0.6cm}

\end{center}
\noindent
\textbf{Abstract.}
We propose an exact polynomial algorithm for a resource allocation problem with convex costs and constraints on partial sums of resource consumptions, in the presence of either continuous or integer variables. No assumption of strict convexity or differentiability is needed.
The method solves a hierarchy of resource allocation subproblems, whose solutions are used to convert constraints on sums of resources into bounds for separate variables at higher levels. The resulting time complexity for the integer problem is $O(n \log m \log (B/n))$, and the complexity of obtaining an $\epsilon$-approximate solution for the continuous case is $O(n \log m  \log (B/\epsilon))$, $n$ being the number of variables, $m$ the number of ascending constraints (such that $m < n$), $\epsilon$ a desired precision, and $B$ the total resource. This algorithm attains the best-known complexity when $m = n$, and improves it when $\log m = o(\log n)$.
Extensive experimental analyses are conducted with four recent algorithms on various continuous problems issued from theory and practice. The proposed method achieves a higher performance than previous algorithms, addressing all problems with up to one million variables in less than one minute on a modern computer.

\vspace*{0.6cm}

\noindent
\textbf{Keywords.}  Separable convex optimization, resource allocation, nested constraints, project crashing, speed optimization, lot sizing

\vspace*{0.6cm}

\noindent
* Corresponding author

\newpage


\section{Problem statement}
\label{s1}
Consider the minimization problem
(\ref{RAP:1}-\ref{RAP:3bis}), with either integer or continuous variables.
Functions $f_i : [0,d_i] \rightarrow  \Re$, $i \in \{1, \dots, n\}$ are proper convex (but not necessarily strictly convex or differentiable). Let $(s[1],\dots,s[m])$ be a subsequence of $m \geq 1$ integers in $\{1,\dots,n\}$ such that $s[m] = n$. The parameters $a_i$, $d_i$ and $B$ are positive integers. For presentation convenience, define $a_0 = 0$, $a_{m} = B$, $s[0] = 0$, $y_i = \sum_{k=1}^{s[i]} x_k$ and $\alpha_i = a_i - a_{i-1}$ for $i \in \{1,\dots,m\}$. 
\begin{align}
\min  \hspace*{0.5cm}  f(\mathbf{x}) =  &\sum\limits_{i=1}^{n}  f_i(x_i) & \label{RAP:1} \\
\text{s.t.}  \hspace*{0.3cm}   \sum\limits_{k=1}^{s[i]}  x_k  \leq & \  a_i   &i  \in \{ 1,\dots,m-1\} \label{RAP:2} \\
 \sum\limits_{i=1}^{n}  x_i =& \ B & \label{RAP:3}  \\
  0 \leq x_i \leq &  \ d_i & i  \in \{ 1,\dots,n\} \label{RAP:3bis}
\end{align}

The problem (\ref{RAP:1}-\ref{RAP:3bis}) appears prominently in a variety of applications related to project crashing \citep{Talbot1982},
production and resource planning \citep{Bellman1954,Bellman1962,Veinott1964}, lot sizing \citep{Tamir1980},
assortment with downward substitution \citep{Hanssmann1957,Sadowski1959,Pentico2008}
departure-time optimization in vehicle routing \citep{Hashimoto2006},
vessel speed optimization \citep{Norstad2010}, 
and telecommunications \citep{Padakandla2009a}, among many others.
When $m = n$, and thus $s[i] = i$ for $i \in \{1,\dots,n\}$ this problem is known as the resource allocation problem with nested constraints (NESTED). In the remainder of this paper, we maintain the same name for any $m \geq 2$. Without the constraints (\ref{RAP:2}), the problem becomes a resource allocation problem (RAP), surveyed in \citep{Ibaraki1988,Patriksson2008}, which is the focus of numerous papers related to search-effort allocation, portfolio selection, energy optimization, sample allocation in stratified sampling, capital budgeting, mass advertising, and matrix balancing, among many others. The RAP can also be solved by  cooperating agents under some mild assumption on functions $f_i$ \citep{Lakshmanan2008}.

In this paper, we proposed efficient polynomial algorithms for both the integer and continuous version of the problem.
Computational complexity concepts are well-defined for linear problems. In contrast, with the exception of seminal works such as \citep{Nemirovsky1983,Monteiro1990,Hochbaum1990,Hochbaum1994} the complexity of algorithms for general non-linear optimization problems is more rarely discussed in the literature, mostly due to the fact that an infinite output size may be needed due to real optimal solutions. To circumvent this issue, we assume the existence of an oracle which returns the value of $f_i(x)$ for any $x$ in a constant number of operations, and rely on an approximate notion of optimality for non-linear optimization problems~\citep{Hochbaum1990}. A solution $\mathbf{x^{(\epsilon)}}$ of a continuous problem is $\epsilon$-accurate if and only there exists an optimal solution $\mathbf{x^*}$ such that $||(\mathbf{x^{(\epsilon)}}-\mathbf{x^*})||_{\infty} \leq \epsilon$. This accuracy is defined in the solution space, in contrast with some other approximation approaches which considered objective space~\citep{Nemirovsky1983}. 

Two main classes of methods can be discerned for NESTED, considering $m=n$.
The two algorithms of Padakandla and Sundaresan \citep{Padakandla2009a} and Wang \citep{Wang2012} can be qualified as \emph{dual}, because they resolve explicitly a succession of RAP sub-problems, assimilated to a single Lagrangian equation, identify and iteratively re-introduce the active nested constraints (\ref{RAP:2}).
These methods attain a complexity of $O(n^2 \Phi_\textsc{Rap}(n,B))$ and $O(n^2 \log n + n\Phi_\textsc{Rap}(n,B))$ for the continuous case, respectively. $\Phi_\textsc{Rap}(n,B)$ is the complexity of solving one RAP with $n$ tasks.  It should be noted that the performance of \citep{Padakandla2009a} can be improved for some specific continuous problems in which the Lagrangian equation admits a closed and additive form. Otherwise, the RAP are solved by a combination of Newton-Raphson and bisection search on the Lagrangian dual. A more precise computational complexity statement for these algorithms would require to describe the complexity of these sub-procedures and the approximation allowed at each step.
Similar methods have also been discussed, albeit with a different terminology, in early production scheduling and lot sizing literature \citep{Modigliani1955}.

Another class of methods, that we classify as \emph{primal}, was primarily designed for the integer version of the problem, but also applies to the continuous case. These methods take inspiration from greedy algorithms, which consider all feasible increments of one resource, and select the least-cost one. The greedy method is known to converge \citep{Federgruen1986} to the optimum of the integer problem when the constraints determining a polymatroid. Dyer and Walker \citep{Dyer1987} thus combine the greedy approach with divide-and-conquer using median search, achieving a complexity of $O(n \log n \log^2 \frac{B}{n})$ in the integer case.
More recently, Hochbaum \citep{Hochbaum1994} combines the greedy algorithm within a scaling approach. An initial problem is solved with large increments, and the increment size is iteratively divided by two to achieve higher accuracy. At each iteration, and for each variable, only one increment from the previous iteration may require to be corrected. Using efficient feasibility checking methods, NESTED can be solved in $O(n \log n \log \frac{B}{n})$. The method can also be applied to the general allocation problem as long as the constraints determine a polymatroid \citep{Hochbaum1994}.

Finally, without constraints (\ref{RAP:2}), the RAP can be solved in $O(n \log \frac{B}{n})$ \citep{Frederickson1982,Hochbaum1994}. This complexity is the best possible \citep{Hochbaum1994} in the comparison model and the algebraic tree model with operations $+,-,\times,\div$.




\section{Contributions}
\label{contrib}

This paper introduces a new algorithm for NESTED, with a complexity of $O(n \log m \log \frac{B}{n})$ in the integer case, and $O(n \log m \log \frac{B}{\epsilon})$ in the continuous case. This is a \emph{dual}-inspired approach, which solves NESTED as a succession of RAP sub-problems as in \citep{Padakandla2009a,Wang2012}. It is the first method of this kind to attain the same best known complexity as \citep{Hochbaum1994} when $m=n$.
In addition, the complexity of the proposed method grows partly with the number of constraints $m$ rather than the number of variables $n$ in \citep{Hochbaum1994}, such that the proposed approach is the fastest known for problems with sparse constraints when $\log m = o(\log n)$.
In the presence of a quadratic objective, the proposed algorithm attains a complexity of $O(n \log m)$, smaller than the previous complexity of $O(n \log n)$ \citep{Hochbaum1995}.

Extensive experimental analyses are conducted to compare our method with previous algorithms, using the same testing environment, on eight problem families with $n$ and $m$ ranging from 10 to 1.000.000.
In practice, the proposed method demonstrates a higher performance than \citep{Hochbaum1994} even when $m = n$, possibly due to the use of very simple data structures. The CPU time is largely smaller than \citep{Padakandla2009a, Wang2012} and the interior point method of MOSEK. All problems with up to one million variables are solved in less than one minute on a modern computer. The method is suitable for large scale problems, e.g. in image processing and telecommunications, or for repeated use when solving combinatorial optimization problems with a resource allocation sub-structure.

Our experiments also show that few nested constraints (\ref{RAP:2}) are usually active in optimal solutions for the considered benchmark instances. In fact, we effectively demonstrate that the expected number of active constraints grows logarithmically with $m$ for some classes of randomly-generated problems. As a corollary, we also highlight a strongly polynomial algorithm for a significant subset of problems.

\section{The proposed algorithm}
\label{algo}

The proposed method performs two simple initialization steps (Algorithm \ref{general-algo}) to guarantee feasibility and then calls upon the main recursive resolution procedure \textsc{Nested}$(1,m)$  (Algorithm~\ref{algo-subprob}). The variables $\mathbf{x},\mathbf{\bar{c}},\mathbf{\bar{d}}$ are global to all functions.
The recursive process results in a hierarchical resolution of NESTED sub-problems, with $1 + \lceil \log m \rceil$ levels of recursion. As will be proven in the following, \textsc{Nested}($v,w$) returns a feasible optimal solution of the sub-problem (\ref{nested1}) associated to the variable range $(x_{s[v-1]+1},\dots,x_{s[w]})$ in the original problem, assuming that nested constraints $v-1$ and $w$ are active,  i.e.  $y_{v-1} = \bar{a}_{v-1}$ and  $y_{w} = \bar{a}_{w}$.

\begin{algorithm}[htbp]
\caption{General solution procedure}
\label{general-algo}
\begin{algorithmic}[1]
\begin{small}
\STATE \textsc{Tightening:}
\STATE  $\bar{a}_0  \gets 0$ ; $\bar{a}_{m} \gets B$
\FOR{$i = 1$ \text{to} $m-1$}
\STATE $\bar{a}_i \gets \min \{\bar{a}_{i-1} + \sum_{k=s[i-1]+1}^{s[i]} d_{k}, \bar{a}_{i} \}$
\ENDFOR
\STATE \textsc{Feasibility:} 
\IF{$\exists \ i \in \{1,\dots,m\}$ such that $\sum_{k=s[i-1]+1}^{n} d_k < B - \bar{a}_{i-1}$}
\STATE \textbf{return Infeasible}
\ENDIF
\STATE \textsc{Hierarchical Resolution:} 
\STATE $(\bar{c}_1,\dots,\bar{c}_n) \gets  (0,\dots,0)$
\STATE $(\bar{d}_1,\dots,\bar{d}_n) \gets  (d_1,\dots,d_n)$
\STATE $(x_1,\dots,x_n) \gets  \textsc{Nested}(1,m)$
\STATE  \textbf{return} $(x_1,\dots,x_n)$
\end{small}
\end{algorithmic}
\end{algorithm}

When $v=w$, the problem \textsc{Nested}$(v,v)$ is a RAP which admits a feasible solution (Lemma \ref{p0}). At each level, an optimal solution to \textsc{Nested}($v,w$) is obtained by recursively solving two subproblems \textsc{Nested}($v,t$) and \textsc{Nested}($t+1,w$) and then solving a modified \textsc{Rap}($v,w$) with an updated range for $\mathbf{x}$.

\begin{algorithm}[htbp]
\caption{$\textsc{Nested}(v,w)$}
\label{algo-subprob}
\begin{algorithmic}[1]
\begin{small}
\IF{$v = w$}
\STATE $(x_{s[v-1]+1},\dots,x_{s[v]}) \gets \textsc{Rap}(v,v)$
\ELSE 
\STATE $t  \gets  \lfloor \frac{v+w}{2} \rfloor$
\STATE $(x_{s[v-1]+1},\dots,x_{s[t]}) \gets \textsc{Nested}(v,t) $
\STATE $(x_{s[t]+1},\dots,x_{s[w]}) \gets \textsc{Nested}(t+1,w) $
\FOR{$i = s[v-1]+1$ to $s[t]$}
\STATE $(\bar{c}_i,\bar{d}_i)  \gets (0,x_i)$
\ENDFOR
\FOR{$i = s[t]+1$ to $s[w]$}
\STATE $(\bar{c}_i,\bar{d}_i)  \gets (x_i,d_i)$
\ENDFOR
\STATE $(x_{s[v-1]+1},\dots,x_{s[w]}) \gets \textsc{Rap}(v,w)$
\ENDIF
\end{small}
\end{algorithmic}
\end{algorithm}

\begin{small}
\begin{align}
\textsc{Nested}(v,w)
&\left\{
\begin{aligned}
\min  \hspace*{0.2cm}  &\sum_{i=s[v-1]+1}^{s[w]}  f_i(x_i) \\
\text{s.t.} \hspace*{0.2cm}  & \sum_{k=s[v-1]+1}^{s[i]}  x_k  \leq \  \bar{a}_i - \bar{a}_{v-1}   && \hspace*{0.5cm}  i  \in \{ v,\dots,w-1\}  \\
& \sum_{i=s[v-1]+1}^{s[w]} x_i = \ \bar{a}_w -  \bar{a}_{v-1}   \\
&  0 \leq x_i \leq  \ d_i &&  \hspace*{0.5cm} i  \in \{ s[v-1]+1,\dots,s[w]\}
\end{aligned} 
\right. \label{nested1} \\
&\nonumber \\ 
\textsc{Rap}(v,w)
&\left\{
\begin{aligned}
\min  \hspace*{0.2cm}  &\sum_{i=s[v-1]+1}^{s[w]}  f_i(x_i) \\
\text{s.t.} \hspace*{0.2cm} &\sum_{i=s[v-1]+1}^{s[w]} x_i = \ \bar{a}_w -  \bar{a}_{v-1}   \\
&  \hat{c}_i \leq x_i \leq  \ \hat{d}_i &&  \hspace*{0.5cm}  i  \in \{ s[v-1]+1,\dots,s[w]\}
\end{aligned} 
\right. \label{rap1}
\end{align}
\end{small}

The rationale of this transformation from NESTED to RAP is explained below.

The optimal solutions $(x^{\downarrow*}_{s[v-1]+1},\dots, x^{\downarrow*}_{s[t]})$ and $(x^{\uparrow*}_{s[t]+1},\dots, x^{\uparrow*}_{s[w]})$ of \textsc{Nested}($v,t$) and \textsc{Nested}($t+1,w$) are such that the $(v-1)^\text{th}$, $t^\text{th}$ and $w^\text{th}$ nested constraints are active ($y_{v-1} = \bar{a}_{v-1}$, $y_{t} =  \bar{a}_t$, $y_{w} =  \bar{a}_w$) and the other nested constraints are satisfied.
Now, the optimal solution $(x^{**}_{s[v-1]+1},\dots, x^{**}_{s[w]})$ of \textsc{Nested}($v,w$) is such that $y_{t} \leq \bar{a}_t$. As a consequence, Theorems \ref{p1} and \ref{p1bis} show that $x^{**}_i \leq x^{\downarrow*}_i$ for $i \in \{s[v-1]+1,\dots,s[t]\}$
and $x^{**}_i \geq x^{\uparrow*}_i$ for $i \in \{s[t]+1,\dots,s[w]\}$.
These valid inequalities can be inserted in the formulation. These inequalities alone also guarantee that nested constraints are satisfied (Corollary \ref{p2}). Nested constraints can thus be eliminated, leading to a \textsc{Rap}$(v,w)$ with updated bounds which can be efficiently solved.
As a consequence, a feasible optimal solution of \textsc{Nested}$(v,w)$ is obtained at each level, leading to an optimal solution of \textsc{Nested}$(1,m)$. \\

\begin{lemma}
\label{p0}
\textsc{Rap}$(v,v)$ admits a feasible solution. 
\end{lemma} 

\begin{proof}
As a consequence of Lines 1 to 8 in Algorithm \ref{general-algo},
$\bar{a}_{v-1} \leq \bar{a}_{v}$ and  $\bar{a}_{v} \leq \bar{a}_{v-1} + \sum_{k=s[v-1]+1}^{s[v]} d_{k}$. 
A feasible solution can then be generated as follows: \vspace*{0.2cm} \\
\hspace*{0.5cm} $\textbf{for} \   i = s[v-1]+1  \ \textbf{to} \ s[v], \ x_i = \min \{d_i,\bar{a}_v - \bar{a}_{v-1} - \sum_{k=s[v-1]+1}^{i-1} x_{k}\}.$
\end{proof}

\begin{theorem}
\label{p1}
Consider $(v,t,w)$ such that $1 \leq v \leq t \leq w \leq m$ and $v < w$.
Let $(x^{\downarrow*}_{s[v-1]+1},\dots, x^{\downarrow*}_{s[t]})$ and $(x^{\uparrow*}_{s[t]+1},\dots, x^{\uparrow*}_{s[w]})$ be optimal solutions of \textsc{Nested}$(v,t)$ and \textsc{Nested}$(t+1,w)$ with integer variables, respectively, then \textsc{Nested}$(v,w)$ with integer variables admits an optimal solution $(x^{**}_{s[v-1]+1},\dots, x^{**}_{s[w]})$ such that 
$x^{**}_i \leq x^{\downarrow*}_i$  for $i \in \{s[v-1]+1,\dots,s[t]\}$ and
 $x^{**}_i \geq x^{\uparrow*}_i$ for $i \in \{s[t]+1,\dots,s[w]\}$. \\
\end{theorem}

\begin{proof}
This proof relies on the optimality of the greedy algorithm for general RAP in the presence of polymatroidal constraints \citep{Federgruen1986,Ibaraki1988}. The greedy algorithm for \textsc{Nested}$(v,w)$ (Algorithm \ref{greedy}) iteratively considers all variables $x_i$ which can be feasibly incremented by one unit, and increments the least-cost one. This algorithm has one degree of freedom in case of tie (Line 4). In the proof, we add a marginal component in the objective function to break these ties in favor of increments that are part of desired optimal solutions.

\begin{algorithm}[!h]
\caption{\textsc{Greedy}}
\label{greedy}
\begin{algorithmic}[1]
\begin{small}
\STATE $\textbf{x} = (x_1,\dots,x_n) \gets  (0,\dots,0)$
\STATE $E \gets (1,\dots,n)$ ; $I \gets  \bar{a}_w -  \bar{a}_{v-1}$
\WHILE{$I > 0$ \textbf{and} $E  \neq \emptyset $}
\STATE Find $i \in E$ such that $f_i(x_{i+1}) - f_i(x_{i}) = \min_{k \in E} \{ f_k(x_{k+1}) - f_k(x_{k}) \}$ 
\STATE $\mathbf{x'} \gets \textbf{x}$ ; $x'_i \gets x_i + 1$. 
\IF{\textbf{x'} is feasible}
\STATE $\textbf{x} \gets \textbf{x'}$ ; $I \gets I-1$
\ELSE 
\STATE  $ E  \gets E \backslash \{i\} $
\ENDIF
\ENDWHILE
\IF{$I > 0$}
\RETURN \textbf{Infeasible}
\ELSE
\RETURN $\mathbf{x}$
\ENDIF
\end{small}
\end{algorithmic}
\end{algorithm}

First, $(x^{\downarrow*}_{s[v-1]+1},\dots, x^{\downarrow*}_{s[t]}, x^{\uparrow*}_{s[t]+1},\dots, x^{\uparrow*}_{s[w]})$ is a feasible solution of \textsc{Nested}$(v,w)$, and thus at least one optimal solution $\widetilde{\mathbf{x}}$ of \textsc{Nested}$(v,w)$ exists. Define $\widetilde{a}_t =  \bar{a}_{v-1} + \sum_{k=s[v-1]+1}^{s[t]}  \widetilde{x}_k$. The feasibility of  $\widetilde{\mathbf{x}}$ leads to
\begin{align}
0 \leq \widetilde{a}_t  &\leq  \bar{a}_t, \\
 \widetilde{a}_t  + \sum_{k=s[t]+1}^{s[w]}  d_k &\geq  \bar{a}_{w}.  \label{greater}
\end{align}

The problem \textsc{Nested}$(v,t)$ has a discrete and finite set of solutions, and the associated set of objective values is discrete and finite. If all feasible solutions are optimal, then set $\xi=1$, otherwise let $\xi > 0$ be the gap between the best and the second best objective value. Consider \textsc{Nested}$(v,t)$ with a modified separable objective function $\mathbf{\bar{f}}$ such that for $i \in \{s[v-1]+1, \dots, s[t]\}$,
\begin{equation}
\bar{f}_i(x) = f_i(x) + \frac{\xi}{B+1}  \max \{x - x^{\downarrow*}_i, 0 \}.
\end{equation}

Any solution $\mathbf{x}$ of the modified \textsc{Nested}$(v,t)$ with $\mathbf{\bar{f}}$ is an optimal solution of the original problem with  $\mathbf{f}$ if and only if $\mathbf{\bar{f}}(\mathbf{x}) < \mathbf{f}(\mathbf{x^{\downarrow*}}) + \xi$, and the new problem admits the unique optimal solution $\mathbf{x^{\downarrow*}}$. Thus, \textsc{Greedy} returns $\mathbf{x^{\downarrow*}}$ after $\bar{a}_t -  \bar{a}_{v-1}$ increments. Let $\mathbf{x}^{**} =  (x^{**}_{s[v-1]+1},\dots, x^{**}_{s[t]})$ be the solution obtained at increment $\widetilde{a}_t -  \bar{a}_{v-1}$. By the properties of  \textsc{Greedy}, $\mathbf{x}^{**}$ is an optimal solution of \textsc{Nested}$(v,t)$ when replacing $\bar{a}_t$ by $\widetilde{a}_t$, such that $x^{**}_i \leq x^{\downarrow*}_i$  for $i \in \{s[v-1]+1,\dots,s[t]\}$. \\

The same process can be used for the subproblem \textsc{Nested}$(t+1,w)$. With the change of variables $\hat{x}_i = d_i - x_i$, and $g_i(x) =   f_i(d_i -  x)$, the problem becomes
\begin{small}
\begin{align}
\textsc{Nested-bis}(t+1,w)
&\left\{
\begin{aligned}
\min  \hspace*{0.2cm}  &\sum_{i=s[t]+1}^{s[w]} g_i(\hat{x}_i)\\
\text{s.t.} \hspace*{0.2cm}  & \sum_{k=s[i]+1}^{s[w]}  \hat{x}_k  \leq \bar{a}_i - \bar{a}_{w}  + \sum_{k=s[i]+1}^{s[w]}  d_k   && \hspace*{0.5cm}  i  \in \{ t+1,\dots,w-1\}  \\
& \sum_{k=s[t]+1}^{s[w]}  \hat{x}_k = \bar{a}_t - \bar{a}_{w}  + \sum_{k=s[t]+1}^{s[w]}  d_k   \\
&  0 \leq \hat{x}_i \leq  \ d_i &&  \hspace*{0.5cm} i  \in \{ s[t]+1,\dots,s[w]\}.
\end{aligned} 
\right. \label{nestedEquiv}
\end{align}
 \end{small}

If all feasible solutions of $\textsc{Nested-bis}(t+1,w)$ are optimal, then set $\hat{\xi}=1$, otherwise
let $\hat{\xi} > 0$ be the gap between the best and second best solution of $\textsc{Nested-bis}(t+1,w)$. For  $i \in \{s[t]+1, \dots, s[w]\}$, define $\hat{x}^{\uparrow*}_i = d_i - x^{\uparrow*}_i$ and $\mathbf{\bar{g}}$ such that
\begin{equation}
\bar{g}_i(x) = g_i(x) + \frac{\hat{\xi}}{B+1}  \max \{x - \hat{x}^{\uparrow*}_i, 0 \}.
\end{equation}

\textsc{Greedy} returns $\mathbf{\hat{x}^{\uparrow*}}$, the unique optimal solution of \textsc{Nested-bis}$(t+1,w)$ with the modified objective $\mathbf{\bar{g}}$.  Let  $\mathbf{\hat{x}^{**}}$ be the solution obtained at step $\widetilde{a}_t - \bar{a}_{w}  + \sum_{k=s[t]+1}^{s[w]}  d_k$. This step is non-negative according to Equation (\ref{greater}).  \textsc{Greedy} guarantees that $\mathbf{\hat{x}^{**}}$ is an optimal solution of \textsc{Nested-bis}$(t+1,w)$ with the alternative equality constraint
\begin{equation}
\sum_{k=s[t]+1}^{s[w]}  \hat{x}_k = \widetilde{a}_t - \bar{a}_{w}  + \sum_{k=s[t]+1}^{s[w]}  d_k.
\end{equation}
In addition,  $\hat{x}^{**}_i \leq \hat{x}^{\uparrow*}_i$  for $i \in \{s[t]+1, \dots, s[w]\}$. Reverting the change of variables, this leads to an optimal solution ${\mathbf{x}}^{**}$ of $\textsc{Nested}(t+1,w)$ where $\bar{a}_t$ has been replaced by $\widetilde{a}_t$, and such that $x^{**}_i \geq x^{\uparrow*}_i$  for $i \in \{s[t]+1, \dots, s[w]\}$.

Overall, since $\mathbf{x}^{**}$ is such that  $\sum_{k=s[v-1]+1}^{s[t]}  x^{**}_k =  \sum_{k=s[v-1]+1}^{s[t]}  \widetilde{x}_k =\widetilde{a}_t - \bar{a}_{v-1}$, since it is also optimal  for the two sub-problems obtained when fixing $\widetilde{a}_t =  \bar{a}_{v-1} + \sum_{k=s[v-1]+1}^{s[t]}  x^{**}_k$, then $x^{**}$ is an optimal solution of  \textsc{Nested}$(v,w)$ which satisfies the requirements of Theorem \ref{p1}.
\end{proof}

\begin{theorem}
\label{p1bis}
The statement of Theorem \ref{p1} is also valid for the problem with continuous variables.
\end{theorem}

\begin{proof}
The proof relies on the proximity theorem of Hochbaum \citep{Hochbaum1994} for general resource allocation problem with polymatroidal constraints. This theorem states that for any optimal continuous solution $\mathbf{x}$ there exists an optimal solution $\mathbf{z}$ of the same problem with integer variables, such that $\mathbf{z} - \mathbf{e} < \mathbf{x} <  \mathbf{z} + n \mathbf{e}$, and thus $||\mathbf{z} - \mathbf{x}||_{\infty} \leq n$. Reversely, for any integer optimal solution $\mathbf{z}$, there exists an optimal continuous solution such that $||\mathbf{z} - \mathbf{x}||_{\infty} \leq n$. \\

Let $(x^{\downarrow*}_{s[v-1]+1},\dots, x^{\downarrow*}_{s[t]})$ and $(x^{\uparrow*}_{s[t]+1},\dots, x^{\uparrow*}_{s[w]})$ be two optimal solutions of \linebreak \textsc{Nested}$(v,t)$ and \textsc{Nested}$(t+1,w)$ with continuous variables, and suppose that the statement of Theorem \ref{p1} is false for the continuous case.
Hence, there exists $\Delta > 0$ such that for any optimal solution $\mathbf{x}^{**}$ of the continuous \textsc{Nested}$(v,w)$ there exists either $i \in \{s[v-1]+1,\dots,s[t]\}$ such that $x^{**}_i \geq  \Delta + x^{\downarrow*}_i$, or $i \in \{s[t]+1,\dots,s[w]\}$ such that $x^{**}_i \leq x^{\uparrow*}_i -  \Delta$. We will prove that this statement is impossible. \\

Define the scaled problem \textsc{Nested-}$\beta(v,t)$ below. This problem admits at least one feasible integer solution as a consequence of the feasibility of \textsc{Nested}$(v,t)$.

\begin{small}
\begin{equation}
\textsc{Nested-}\beta(v,t)
\left\{
\begin{aligned}
\min  \hspace*{0.2cm}  &\sum_{i=s[v-1]+1}^{s[t]}  f_i \left( \frac{x_i}{\beta} \right) \\
\text{s.t.} \hspace*{0.2cm}  & \sum_{k=s[v-1]+1}^{s[i]}  x_k  \leq \beta \bar{a}_i - \beta  \bar{a}_{v-1} && \hspace*{0.5cm}  i  \in \{ v,\dots,t-1\}  \\
& \sum_{i=s[v-1]+1}^{s[t]} x_i  = \beta  \bar{a}_t -  \beta  \bar{a}_{v-1} \\
&  0 \leq x_i \leq  \    \beta d_i  &&  \hspace*{0.5cm} i  \in \{ s[v-1]+1,\dots,s[t]\}
\end{aligned} 
\right. \label{nestedS}
\end{equation}
\end{small}

The proximity theorem of \citep{Hochbaum1994} guarantees the existence of a serie of integer solutions $\mathbf{\hat{x}}^{\downarrow*[\beta]}$ of $\textsc{Nested-}\beta(v,t)$ such that $\lim_{\beta \rightarrow \infty} \Vert  \frac{\mathbf{\hat{x}}^{\downarrow*[\beta]}}{\beta}  - \mathbf{x}^{\downarrow*} \Vert = 0$. With the same arguments, the existence of a serie of integer solutions $\mathbf{\hat{x}}^{\uparrow*[\beta]}$ of $\textsc{Nested-}\beta(t+1,w)$ such that $\lim_{\beta \rightarrow \infty} \Vert \frac{\mathbf{\hat{x}}^{\uparrow*[\beta]}}{\beta}  - \mathbf{x}^{\uparrow*} \Vert = 0$ is also demonstrated.

As a consequence of Theorem \ref{p1}, for any $\beta$ there exists an integer optimal solution $\mathbf{\hat{x}}^{**[\beta]}$ of $\textsc{Nested-}\beta(v,w)$, such that $\hat{x}^{**[\beta]}_i \leq \hat{x}^{\downarrow*[\beta]}_i$ for $i \in \{s[v-1]+1,\dots,s[t]\}$ and
 $\hat{x}^{**[\beta]}_i \geq \hat{x}^{\uparrow*[\beta]}_i$ for $i \in \{s[t]+1,\dots,s[w]\}$.

Finally, the proximity theorem of \citep{Hochbaum1994} guarantees the existence of continuous solutions $\mathbf{x}^{**[\beta]}$ of $\textsc{Nested-}\beta(v,w)$, such that  $\lim_{\beta \rightarrow \infty} \Vert \mathbf{x}^{**[\beta]} - \frac{\mathbf{\hat{x}}^{**[\beta]}}{\beta} \Vert = 0$. \\

Hence, there exist $\beta$, $\mathbf{\hat{x}}^{\downarrow*[\beta]}$, $\mathbf{\hat{x}}^{\uparrow*[\beta]}$,  $\mathbf{\hat{x}}^{**[\beta]}$ and $\mathbf{x}^{**[\beta]}$ such that   
$\Vert \frac{\mathbf{\hat{x}}^{\downarrow*[\beta]}}{\beta} - \mathbf{x}^{\downarrow*} \Vert \leq \frac{\Delta}{3}$, 
$\Vert \frac{\mathbf{\hat{x}}^{\uparrow*[\beta]}}{\beta} - \mathbf{x}^{\uparrow*} \Vert \leq  \frac{\Delta}{3}$, 
and $\Vert \frac{\mathbf{\hat{x}}^{**[\beta]}}{\beta} - \mathbf{x}^{**[\beta]} \Vert \leq  \frac{\Delta}{3}$.

For $i \in \{s[v-1]+1,\dots,s[t]\}$, we have 
$x^{**[\beta]}_i  
\leq \frac{\hat{x}^{**[\beta]}_i}{\beta} + \frac{\Delta}{3}
\leq \frac{\hat{x}^{\downarrow*[\beta]}_i}{\beta}  + \frac{\Delta}{3}
\leq x^{\downarrow*}_i + \frac{2\Delta}{3}$. As a consequence, the statement $x^{**[\beta]}_i \geq \Delta + x^{\downarrow*}_i$ is false.

For $i \in \{s[t]+1,\dots,s[w]\}$, we have
$x^{**[\beta]}_i 
\geq \frac{\hat{x}^{**[\beta]}_i}{\beta} - \frac{\Delta}{3}
\geq \frac{\bar{x}^{\uparrow*[\beta]}_i}{\beta}  - \frac{\Delta}{3}
\geq x^{\uparrow*}_i - \frac{2\Delta}{3}$. As a consequence, the statement $x^{**[\beta]}_i \leq x^{\uparrow*}_i -  \Delta$ is false, and the solution $x^{**[\beta]}$ leads to the announced contradiction. 
\end{proof}

\begin{corollary}
\label{p2}
Let $(x^{\downarrow*}_{s[v-1]+1},\dots, x^{\downarrow*}_{s[t]})$ and $(x^{\uparrow*}_{s[t]+1},\dots, x^{\uparrow*}_{s[w]})$ be two optimal solutions of \textsc{Nested}$(v,t)$ and \textsc{Nested}$(t+1,w)$, respectively. Then, \textsc{Rap}$(v,w)$ with the coefficients $\mathbf{\hat{c}}$ and $\mathbf{\hat{d}}$ given below admits at least one optimal solution, and any of its optimal solutions is also an optimal solution of \textsc{Nested}$(v,w)$. This proposition is valid for continuous and integer variables.
\begin{equation*}
\hat{c}_i = 
\begin{cases}
0 & i \in \{s[v-1]+1,\dots,s[t]\} \\
x_i^* & \text{otherwise}
\end{cases} \hspace*{0.5cm} \text{and}  \hspace*{0.5cm}
\hat{d}_i =
\begin{cases}
 x_i^* & i \in \{s[t]+1,\dots,s[w]\} \\
 d_i & \text{otherwise}
\end{cases}
\end{equation*}
\end{corollary}

\begin{proof}
As demonstrated in Theorems \ref{p1} and \ref{p1bis}, there exists an optimal solution $\textbf{x}^{**}$ of \textsc{Nested}$(v,w)$, such that $x^{**}_k \leq x^{\downarrow*}_k$  for $k \in \{s[v-1]+1,\dots,s[t]\}$ and $x^{**}_k \geq x^{\uparrow*}_k$ for $k \in \{s[t]+1,\dots,s[w]\}$. These two sets of constraints can be introduced in the formulation (\ref{nested1}). Any optimal solution of this strengthened formulation is an optimal solution of \textsc{Nested}$(v,w)$, and the strengthened formulation admits at least one feasible solution. The following relations hold for any solution $\mathbf{x} = (x_{s[v-1]+1},\dots,x_{s[w]})$:

\begin{equation}
\begin{aligned}
x_k \leq x^{\downarrow*}_k \text{ for }   & k  \in \{ s[v-1]+1,\dots,s[t]\}  \\ &\Rightarrow  \sum_{k=s[v-1]+1}^{s[i]} x_k  \leq   \sum_{k=s[v-1]+1}^{s[i]} x^{\downarrow*}_k  \\ &\Rightarrow  \sum_{k=s[v-1]+1}^{s[i]} x_k  \leq \bar{a}_i - \bar{a}_{v-1}  \text{ for }  i  \in \{ v,\dots,t\} 
\end{aligned}
\end{equation}
\begin{equation}
\begin{aligned}
 x_k \geq x^{\uparrow*}_k \text{ for }  & k  \in \{ s[t]+1,\dots,s[w]\} \\ &\Rightarrow  \sum_{k=s[i]+1}^{s[w]} x_k \geq  \sum_{k=s[i]+1}^{s[w]} x^{\uparrow*}_k \\ 
&\Rightarrow  \sum_{k=s[v-1]+1}^{s[i]} x_k \leq  \sum_{k=s[v-1]+1}^{s[i]}  x^{\uparrow*}_k \\
&\Rightarrow   \sum_{k=s[v-1]+1}^{s[i]}  x_k \leq   \bar{a}_{i} - \bar{a}_{v-1}  \text{ for }   i  \in \{ t,\dots,w-1\} 
\end{aligned}
\end{equation}

Hence, any solution satisfying the constraints $x_i \leq \hat{d}_i$ for $i \in \{s[v-1]+1,\dots,s[t]\}$ and $x^{**}_i \geq \hat{c}_i$ for $i \in \{s[t]+1,\dots,s[w]\}$ also satisfies the  constraints 

\begin{equation}
 \sum_{k=s[v-1]+1}^{s[i]}  x_k  \leq \  \bar{a}_i - \bar{a}_{v-1}  \text{ for }   i  \in \{ v,\dots,w-1\}. \label{useless}
 \end{equation}
The nested constraints (\ref{useless}) can thus be removed, and the formulation \textsc{Rap}$(v,w)$ is obtained.
\end{proof}

\section{Computational complexity}
\label{complexity}

This section investigates the computational complexity of the proposed method for integer and continuous problems, as well as for the specific case of quadratic objective functions.  \\

\begin{theorem}
\label{p4}
The proposed algorithm for NESTED with integer variables works with a complexity of $O(n \log n \log \frac{B}{n})$.
\end{theorem}

\begin{proof}
After a pre-processing step in $O(n)$ operations (Algorithm 1, Lines 1 to 7), the integer NESTED problem is solved as a hierarchy of RAP, with $h = 1 +  \lceil  \log_2 m  \rceil $ levels of recursion (Algorithm 2, Lines 4 to 6). At each level $i \in \{1,\dots,h\}$, $2^{h-i}$ RAP sub-problems are solved (Algorithm 2, Lines 2 and 11). Furthermore, there are $O(n)$ operations per level to maintain $\mathbf{\hat{c}}$ and $\mathbf{\hat{d}}$ (Algorithm 2, L7-10).
The method of Frederickson and Johnson \citep{Frederickson1982} for RAP works in $O(\log n \log \frac{B}{n})$. Hence, each $\textsc{Rap}(v,w)$ can be solved in $O((s[w]-s[v]) \log \frac{a_{w} - a_{v}}{s[w]-s[v]})$ operations.
Overall, there exist positive constants $K$, $K'$ and $K''$ such that the number of operations $\Phi(n,m)$ of the proposed method is

\begin{small}
\begin{equation*}
\begin{aligned}
\Phi(n,m,B) &\leq K n + \sum_{i=1}^{h} \left( K' n +  \sum_{j=1}^{2^{h-i}} 
 K''\left( s[2^i j] - s[2^i (j-1)] \right) \log{\left( \frac{a_{2^i \times j} - a_{2^i \times (j-1)} } {s[2^i j] - s[2^i (j-1)]} \right)} \right)  \\
&= K n + K' n h +  K'' n \sum_{i=1}^{h} \sum_{j=1}^{2^{h-i}} 
\frac{s[2^i j] - s[2^i (j-1)]}{n} \log {\left(\frac{a_{2^i \times j} - a_{2^i \times (j-1)} } {s[2^i j] - s[2^i (j-1)]}\right)} \end{aligned}
\end{equation*}
\begin{equation*}
\begin{aligned}
&\leq K n + K' n h +  K'' n \sum_{i=1}^{h} \log{  \left(   \frac{\sum_{j=1}^{2^{h-i}}   (a_{2^i \times j} - a_{2^i \times (j-1)} )}{n} \right)  } \\
&\leq K n + K' n h +  K'' n h \log \frac{B}{n}\\
& = K n + K' n (1+   \lceil  \log m  \rceil ) +  K'' n (1+   \lceil  \log m  \rceil ) \log \frac{B}{n}.
\end{aligned}
\end{equation*}
\end{small}

This leads to the announced complexity of $O(n \log m \log \frac{B}{n})$. 
\end{proof}

For the continuous case, two situations can be discerned.
When there exists an ``exact'' solution method independent of $\epsilon$ to solve the RAP subproblems, e.g. when the objective function is quadratic, the convergence is guaranteed by Theorem \ref{p1bis}.
As such, the algorithm of Brucker \citep{Brucker1984} or Maculan et al. \citep{Maculan2003} can be used to solve each quadratic RAP sub-problem in $O(n)$, leading to an overall complexity of $O(n \log m)$ to solve the quadratic NESTED resource allocation problem.

In the more general continuous case without any other assumption on the objective functions, all problem parameters can be scaled by a factor $\frac{n}{\epsilon}$ \citep{Hochbaum1994}, and the integer problem with $B' = \frac{Bn}{\epsilon}$ can be be solved with complexity $O(n \log m \log \frac{B'}{n \epsilon}) \allowbreak = O(n \log m \log \frac{B}{\epsilon})$. The proximity theorem guarantees that an $\epsilon$-accurate solution of the continuous problem is obtained after the reverse change of variables. \\

Finally, we have assumed in this paper integer values for $a_i$, $B$, and $d_i$.
Now, consider fractional parameter values with $z$ significant figures and $x$ decimal places. All problem coefficients as well as $\epsilon$ can be scaled by a factor $10^x$ to obtain integer parameters, and the number of elementary operations of the method remains the same. We assume in this process that operations are still elementary for numbers of $z+x$ digits. This is a common assumption when dealing with continuous parameters.

\section{Experimental analyses}
\label{exp}

To assess the practical performance of the proposed algorithm, we implemented it as well as the three other methods. In the following, we thus compare
\begin{itemize}[nosep]
\item PS09 : the dual algorithm of Padakandla and Sundaresan \citep{Padakandla2009};
\item W14 : the dual algorithm of Wang \citep{Wang2012};
\item H94 : the scaled greedy algorithm of Hochbaum \citep{Hochbaum1994};
\item MOSEK : the interior point method of MOSEK \citep[][for conic quadratic opt.]{Andersen2003};
\item THIS : The proposed method.  \\
\end{itemize}

Each algorithm is tested on NESTED instances with three types of objective functions.
The first objective function profile comes from \citep{Padakandla2009,Wang2012}. We also consider two other objectives related to project and production scheduling applications. The size of instances ranges from $n=10$ to $1,000,000$. To further investigate the impact of the number of nested constraints, additional instances with a fixed number of tasks and a variable number of nested constraints are also considered. An accuracy of $\epsilon = 10^{-8}$ is sought, and all tests are conducted on a Xeon 3.07 GHz CPU.

The instances proposed in \citep{Padakandla2009,Wang2012} are continuous with non-integer parameters $a_i$ and $B$. We generated these parameters with nine decimals. Following the last remark of Section \ref{complexity}, all problem parameters and $\epsilon$ can be multiplied by $10^{9}$ to obtain a problem with integer coefficients. For a fair comparison with previous authors, we rely on a similar RAP method as \citep{Srikantan1963,Padakandla2009,Wang2012}, using bisection search on the Lagrangian equation to solve the sub-problems. The derivative $f'_i$ of $f_i$ is well-defined for all test instances. This Lagrangian method does not have the same complexity guarantees as \citep{Frederickson1982}, but performs reasonably well in practice. The initial bisection-search interval for each $\textsc{Rap}(v,w)$ is set to
$[ \min_{i\in\{s[v-1]+1,\dots,s[w]\}} f'_i(\hat{c_i}),   \max_{i\in\{s[v-1]+1,\dots,s[w]\}}  f'_i(\hat{d}_i)]$.
 
Implementing  previous algorithm from the literature led to some further questions and implementation choices.  As mentioned in \citep{Hochbaum1994}, some specific implementations of Union-Find \citep{Gabow1985} can achieve a $O(1)$ amortized complexity for feasibility checks. The implementation of these dedicated structures is intricate, and we privileged a more standard Union-Find with  balancing and path compression \citep{Tarjan1975}, attaining a complexity of $\alpha^{\textsc{ack}}(n)$ where $\alpha^{\textsc{ack}}$ is the inverse of the Ackermann function. For all practical purposes, this complexity is nearly constant and the practical performance of the simpler implementation is very similar. The algorithm of \citep{Hochbaum1994} also requires a correction, which is well-documented in \citep{Moriguchi2004}.
In \citep{Wang2012}, the first positive value of the function of \citep[][Equation 20]{Wang2012} is stated to be obtained by binary search. However, this function is non-monotonous in practice.
Looking further at this equation, we observed that only indices $r$ related to current active constraints $a_r$ should be considered. With this change, the function becomes monotonous and binary search can be applied.

Finally, as observed in our experiments, most randomly-generated instances lead to solutions with few active nested constraints. This aspect is discussed in Section \ref{active}. To  further investigate the performance of all methods on different settings, an alternative parameter generation process has been used to produced additional instances.

\subsection{Problem instances -- previous literature}
\label{applications}

We first consider the test function (\ref{f2}) from \citep{Padakandla2009}. The problem was originally formulated with nested constraints of the type $\sum_{k=1}^{s[i]}  x_k  \geq  a_i $ for $i  \in \{ 1,\dots,m-1\}$.  The change of variables $\hat{x_i} = 1- x_i$ can be used to obtain (\ref{RAP:1}-\ref{RAP:3bis}).
\begin{align}
[\text{F}] \hspace*{1cm}  f_i (x) &= \frac{x^4}{4} + p_i x,  & x \in [0,1]  \label{f2}
\end{align}

The problem instances of \citep{Padakandla2009} have been generated with uniformly distributed $p_i$ and $\alpha_i$ in [0,1] (recall that $\alpha_i = a_i - a_{i-1}$). The $p_i$ are then sorted by increasing value. As observed in our experiments, this ordering of parameters leads to very few active nested constraints. To examine method performances on a wider range of settings, we thus generated two other instance sets called [F-Uniform] and [F-Active]. In [F-Uniform], parameters $p_i$ and $\alpha_i$ are generated with uniform distribution, between [0,1] and [0,0.5], respectively, and non-ordered. [F-Active] is generated in the same way, and $\alpha_i$ are sorted in decreasing order. As a consequence, these latter instances have many active constraints.  \citep{Padakandla2009} considered some other test functions for which the solution of the Lagrangian equation admits a closed additive form. As such, each Lagrangian equation can be solved in PS09 in amortized $O(1)$ instead of $O(n \log \frac{B}{n})$. This case is very specific and does not take into account arbitrary bounds $d_i$. Thus, we selected the first type of function for our experiments since it is representative of the general case and does not open the way to function-specific strategies.

\subsection{Problem instances -- Project crashing} 
\label{R-PERT}
A seminal problem in project management \citep{Kelley1959,Malcolm1959} relates to the optimization of a critical path of tasks in the presence of non-linear cost/time trade-off functions $f_i (x)$, expressing the cost of processing a task $i$ in $x_i$ time units. Different types of trade-off functions have been investigated in the literature \citep{Fulkerson1961,Berman1964,Robinson1975,Foldes1993,Cheng1998}. The algorithm of this paper can provide the best compression of a critical path to finish a project at time $B$, while imposing additional deadline constraints $\sum_{k=1}^{s[i]}  x_k  \leq  a_i $ for $i  \in \{ 1,\dots,m-1\}$ on some steps of the project. Lower and upper bounds on task durations $c_i \leq x_i \leq d_i$ are also commonly imposed. The change of variables $\hat{x}_i = x_i + c_i$ leads to the formulation (\ref{RAP:1}-\ref{RAP:3bis}).
Computational experiments are performed on these problems with the cost/time trade-off functions of Equation (\ref{cost-pert}), proposed in \citep{Foldes1993}, in which the cost supplement related to crashing grows as the inverse of task duration. 
\begin{align}
[\text{Crashing}] \hspace*{1cm} f_i(x) = k_i +  \frac{p_i}{x}, &\hspace{0.5cm} x  \in [c_i,d_i]  \label{cost-pert}
\end{align}
Parameters $p_i$, $d_i$ and $\alpha_i$ are generated by exponential distributions of mean $\mathbb{E}(p_i) = \mathbb{E}(d_i) = 1$ and $\mathbb{E}(\alpha_i) = 0.75$. Finally, $a_i = \sum_{k=1}^i \alpha_k$ and $c_i = \min (\alpha_i,\frac{d_i}{2})$ to ensure feasibility.

\subsection{Problem instances -- Vessel speed optimization} Some applications require solving multiple NESTED problems. One such case relates to an emergent class of vehicle routing and scheduling problems aiming at jointly optimizing vehicle speeds and routes to reach delivery locations within specified time intervals \citep{Norstad2010,Bektas2011a,Hvattum2013}.
Heuristic and exact methods for such problems consider a very large number of alternative routes (permutations of visits) during the search. For each route, determining the optimal travel times $(x_1,\dots,x_n)$ on $n$ trip segments to satisfy $m$ deadlines $(a_1, \dots, a_m)$ on some locations is the same subproblem as in formulation (\ref{RAP:1}-\ref{RAP:3bis}).
We generate a set of benchmark instances for this problem, assuming as in \citep{Ronen1982} that fuel consumption is approximately a cubic function of speed on relevant intervals. In Equation (\ref{vessel-fuel}), $p_i$ is the fuel consumption on the way to location $i$ per time unit at maximum speed, and $c_i$ is the minimum travel time.
\begin{align}
[\text{FuelOpt}] \hspace*{1cm} f_i(x) = p_i \times c_i \times \left(  \frac{c_i}{x} \right)^3, &\hspace{0.5cm} x  \in [c_i,d_i]  \label{vessel-fuel}
\end{align}

Previous works on the topic \citep{Norstad2010,Hvattum2013} assumed identical $p_i$ on all edges. Our work allows to raise this simplifying assumption, allowing to take into consideration edge-dependent factors such as currents, water depth, or wind which have a strong impact on fuel consumption. We generate uniform $p_i$ values in the interval $[0.8,1.2]$. Base travel times $c_i$ are generated with uniform distribution in $[0.7,1]$, $d_i = c_i * 1.5$, and $\alpha_i$ are generated in $[1,1.2]$.

\begingroup
\renewcommand{\arraystretch}{0.93}
\begin{table}[htbp]
\begin{small}
\begin{center}
\caption{CPU time(s) of five algorithms for NESTED, with increasing $n$ and $m=n$}
\label{detailed-n}
\scalebox{0.79}
{
\setlength{\tabcolsep}{2.7mm}
\begin{tabular}{|cc|@{\hspace*{0.15cm}} c@{\hspace*{0.15cm}}|ccccc|}
\hline
\multirow{2}{*}{\textbf{Instance}}&\multirow{2}{*}{\textbf{n}}&\multirow{2}{*}{\textbf{nb Active}}&\multicolumn{5}{c|}{\textbf{Time (s)}}\\
&&&\textbf{PS09}&\textbf{W14}&\textbf{H94}&\textbf{MOSEK}&\textbf{THIS}\\
\hline
&&&&&&& \vspace*{-0.25cm}\\
\textbf{[F]}&10&1.15&\num{0.00008859}&\num{0.00008056}&\num{0.00006176}&\num{0.0087285}&\textbf{\num{0.0000185}}\\
&10$^2$&1.04&\num{0.00795811}&\num{0.00703437}&\num{0.00067427}&\num{0.02031415}&\textbf{\num{0.00016851}}\\
&10$^3$&1.08&\num{0.91694619}&\num{0.78660605}&\num{0.00873718}&\num{9.6272}&\textbf{\num{0.0019801}}\\
&10$^4$&1.15&\num{105.8343}&\num{87.2204}&\num{0.14593618}&--&\textbf{\num{0.02234682}}\\
&10$^5$&1.20&--&--&\num{2.9307}&--&\textbf{\num{0.3671207}}\\
&10$^6$&1.10&--&--&\num{44.237}&--&\textbf{\num{4.3585}}\\
&&&&&&& \vspace*{-0.1cm}\\
\textbf{[F-Uniform]}&10&2.92&\num{0.00010316}&\num{0.00004571}&\num{0.00005855}&\num{0.00876494}&\textbf{\num{0.00002621}}\\
&10$^2$&5.06&\num{0.01369624}&\num{0.0016079}&\num{0.00074224}&\num{0.02141905}&\textbf{\num{0.00049744}}\\
&10$^3$&7.65&\num{2.28287243}&\num{0.08349505}&\num{0.00983498}&\num{8.6303}&\textbf{\num{0.00841483}}\\
&10$^4$&9.99&--&\num{6.0832}&\num{0.16711636}&--&\textbf{\num{0.13123479}}\\
&10$^5$&12.00&--&--&\num{3,99358333}&--&\textbf{\num{2,73818333}}\\
&10$^6$&14.50&--&--&\num{70.648}&--&\textbf{\num{46.21}}\\
&&&&&&& \vspace*{-0.1cm}\\
\textbf{[F-Active]}&10&3.67&\num{0.00011942}&\num{0.00003938}&\num{0.0000576}&\num{0.00870827}&\textbf{\num{0.00002884}}\\
&10$^2$&10.00&\num{0.02279989}&\num{0.00096467}&\num{0.00074972}&\num{0.02183788}&\textbf{\num{0.00046875}}\\
&10$^3$&22.58&\num{4.88150833}&\num{0.03817105}&\num{0.00992958}&\num{10.12315}&\textbf{\num{0.00681416}}\\
&10$^4$&50.75&--&\num{2.30688}&\num{0.16246152}&--&\textbf{\num{0.09945495}}\\
&10$^5$&114.50&--&\num{262.02}&\num{3.17895}&--&\textbf{\num{1.46951389}}\\
&10$^6$&280.30&--&--&\num{56.54}&--&\textbf{\num{22.138}}\\
&&&&&&& \vspace*{-0.1cm}\\
\textbf{[Crashing]}&10&6.44&\num{0.00004487}&\num{0.00001808}&\num{0.00005022}&\num{0.00945604}&\textbf{\num{0.000008}}\\
&10$^2$&24.61&\num{0.00603407}&\num{0.00070544}&\num{0.00068045}&\num{0.05951631}&\textbf{\num{0.00012454}}\\
&10$^3$&34.14&\num{1.0958892}&\num{0.04843975}&\num{0.00886037}&\num{14.28322809}&\textbf{\num{0.00247863}}\\
&10$^4$&46.90&\num{249.678}&\num{2.84608988}&\num{0.14983694}&--&\textbf{\num{0.04929541}}\\
&10$^5$&50.30&--&\num{298.146}&\num{3.44475}&--&\textbf{\num{1.12707155}}\\
&10$^6$&88.30&--&--&\num{60.211}&--&\textbf{\num{23.464}}\\
&&&&&&& \vspace*{-0.1cm}\\
\textbf{[FuelOpt]}&10&2.93&\num{0.00008462}&\num{0.00003168}&\num{0.00006623}&\num{0.00873563}&\textbf{\num{0.00002197}}\\
&10$^2$&5.31&\num{0.01216929}&\num{0.0012775}&\num{0.00079803}&\num{0.01985813}&\textbf{\num{0.00042075}}\\
&10$^3$&6.86&\num{1.73764873}&\num{0.07099825}&\num{0.01067473}&\num{7.01935}&\textbf{\num{0.00683301}}\\
&10$^4$&9.53&\num{243.3978}&\num{4.812525}&\num{0.19470302}&--&\textbf{\num{0.10156978}}\\
&10$^5$&14.90&--&\num{433.962}&\num{4.87966667}&--&\textbf{\num{1.7188631}}\\
&10$^6$&12.80&--&--&\num{85.389}&--&\textbf{\num{29.91}}\\
\hline
\end{tabular}
}
\end{center}
\end{small}
\end{table}
\endgroup 

\subsection{Experiments with $\mathbf{m = n}$}

The first set of experiments involves as many nested constraints as variables ($n=m$).
We tested the five methods for $n \in  \{10,20,50,\allowbreak100,200,\dots,10^6\}$, with 100 different problem instances for each size $n \leq 10,000$, and 10 different problem instances when $n > 10,000$. A time limit of 10 minutes per run was imposed.
The CPU time of each method for a subset of size values is reported in Table \ref{detailed-n}. The first two columns report the instance set identifier, the next column displays the average number of active constraints in the optimal solutions, and the five next columns report the average run time of each method on each set. The smallest CPU time is highlighted in boldface.
A sign ``--'' means that the time limit is attained without returning a solution.
The complete results, for all values of $n$, are also represented on a logarithmic scale in Figure \ref{figure-n}.

\begin{figure}[htbp]
\hspace*{-0.8cm}
 \begin{minipage}[c]{1.07\textwidth}
\centering
 \begin{minipage}[l]{0.49\textwidth}
\centering
\includegraphics[width=\textwidth]{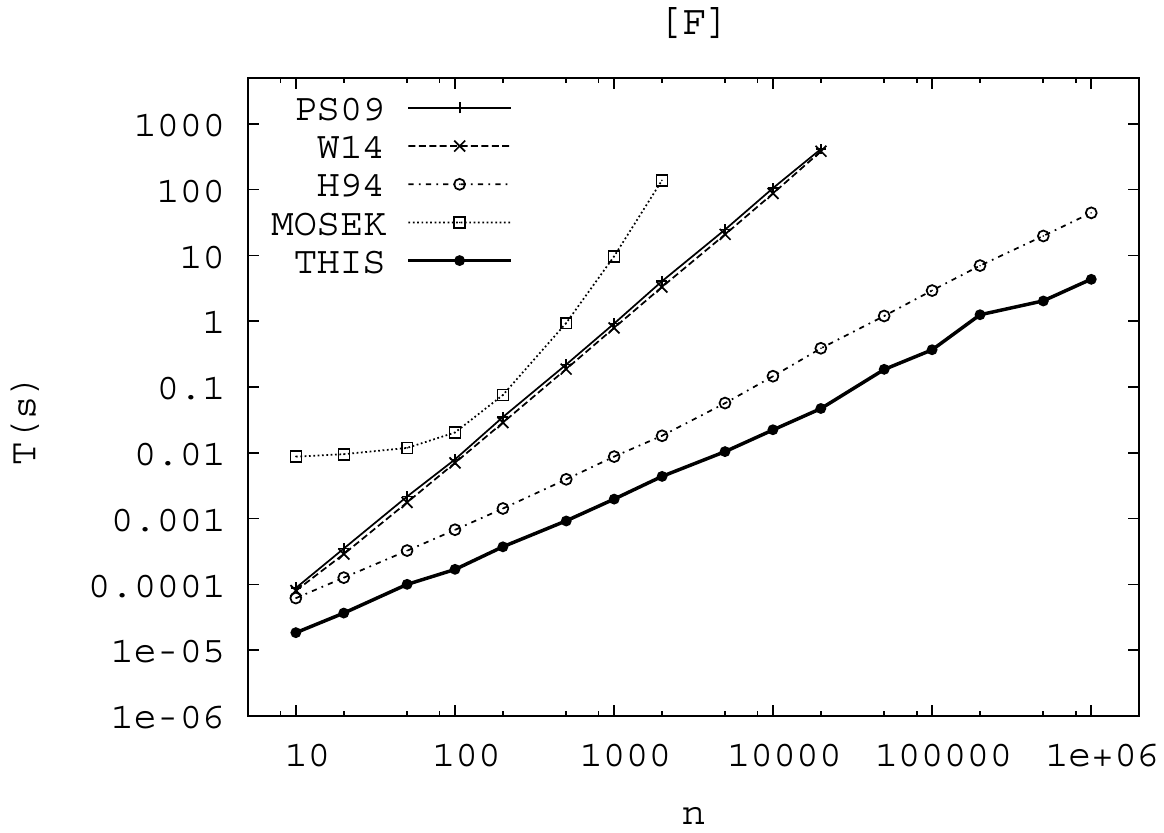}

\vspace*{0.2cm}

\includegraphics[width=\textwidth]{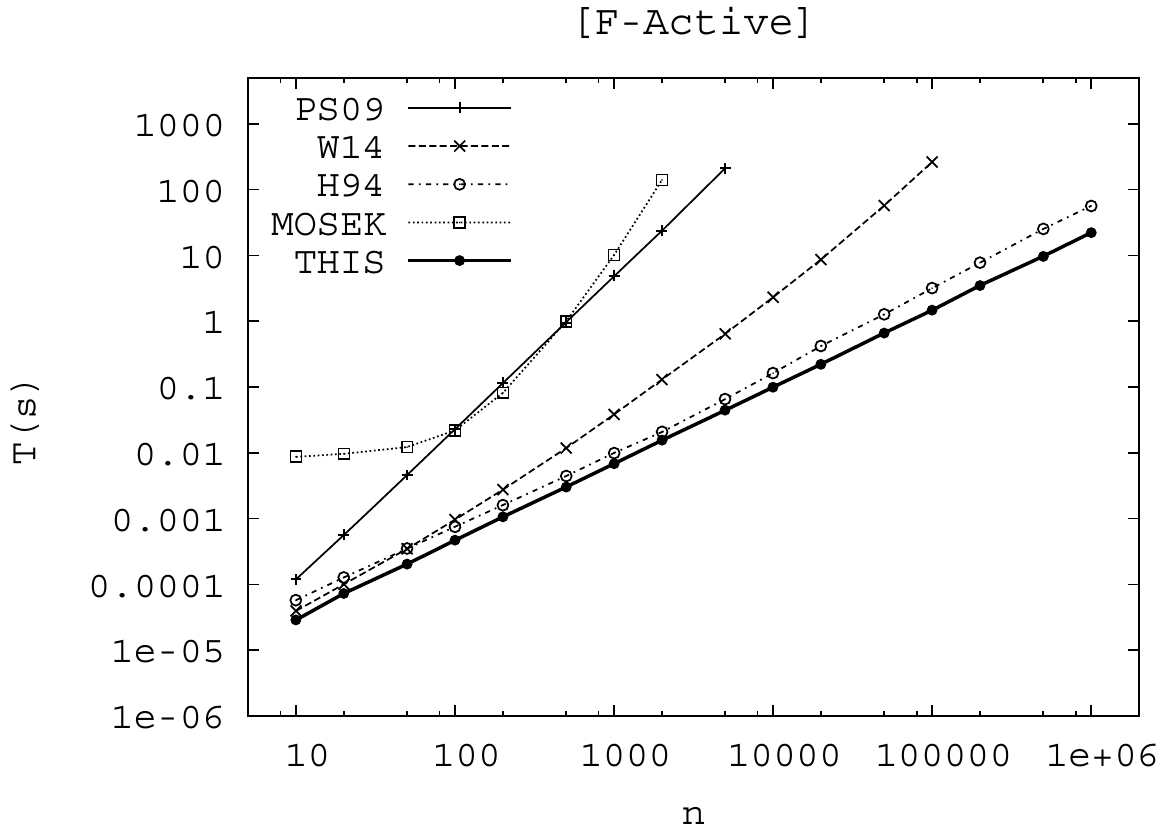}

 \end{minipage}
 \begin{minipage}[r]{0.49\textwidth}
\centering
\includegraphics[width=\textwidth]{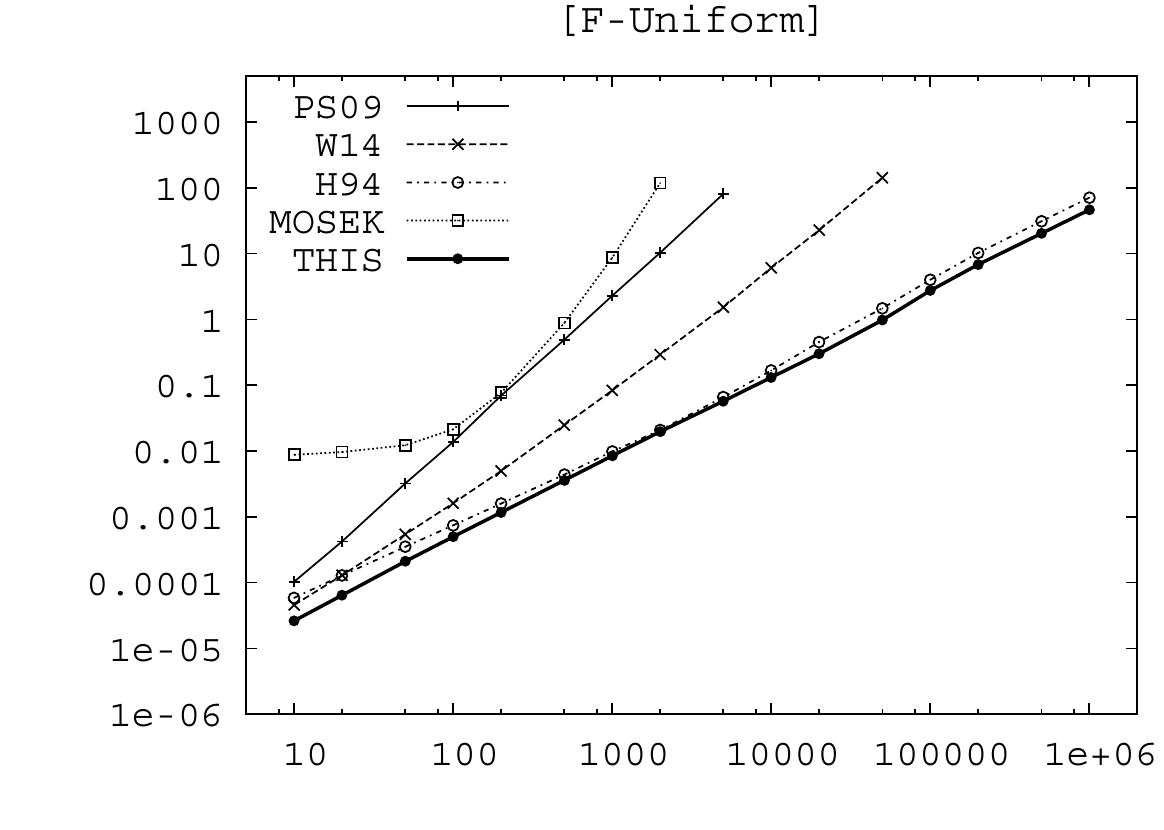}

\vspace*{0.2cm}

\includegraphics[width=\textwidth]{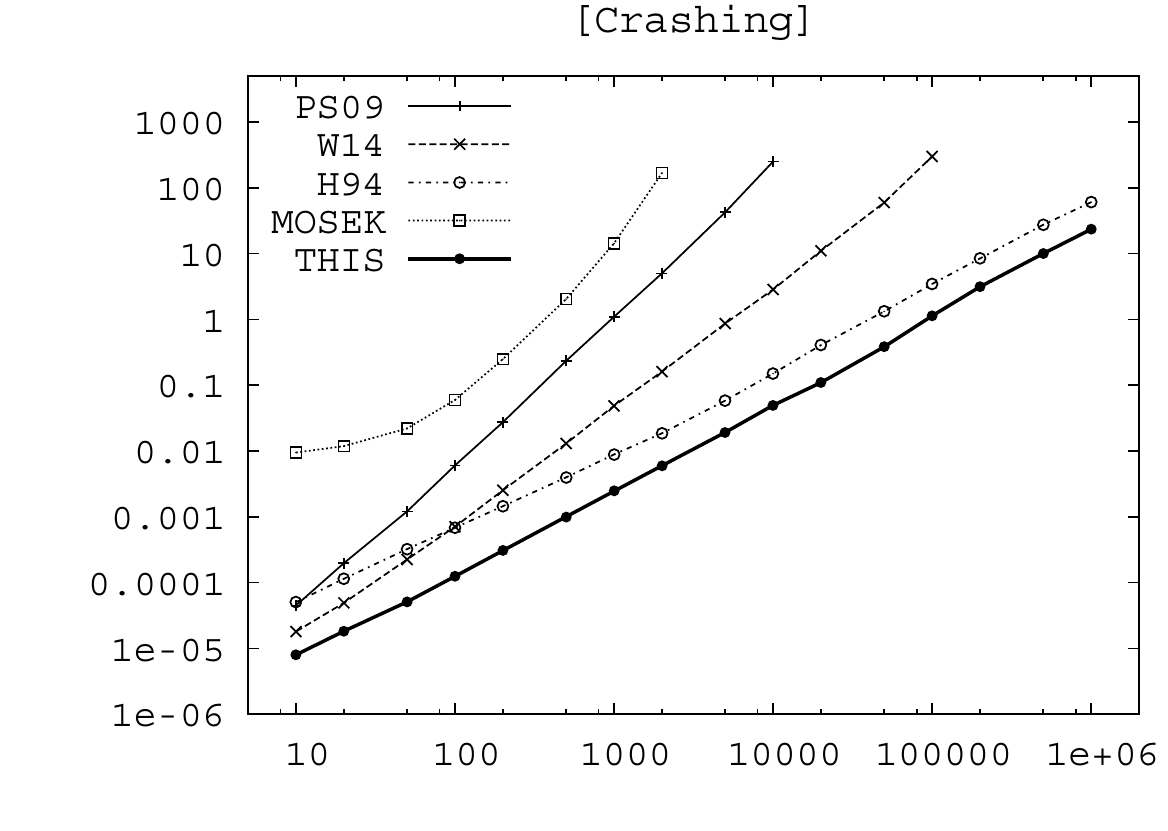}
\end{minipage} 

\vspace*{0.2cm}

\includegraphics[width=0.49\textwidth]{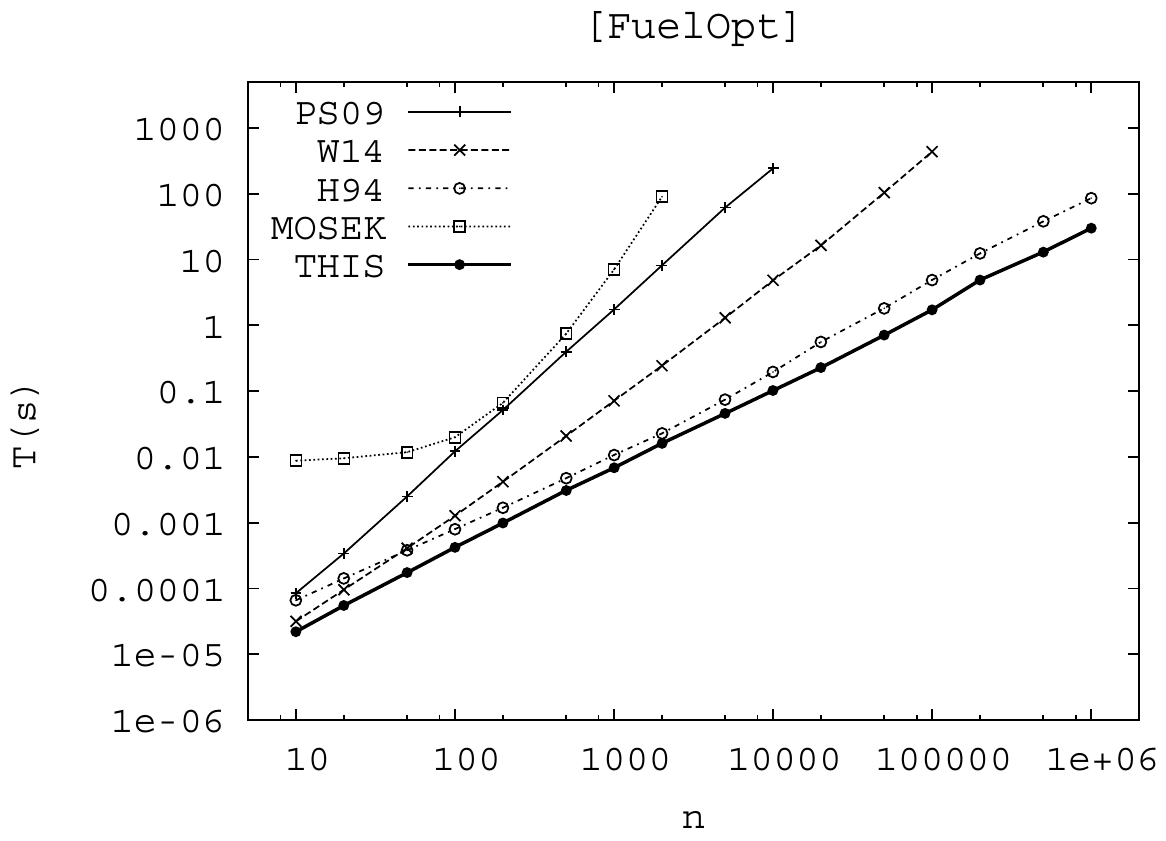}
\end{minipage} 
\caption{CPU Time(s) as a function of $n \in \{10,\dots,10^6\}$. $m=n$. Logarithmic representation}
\label{figure-n}
\end{figure}

\begin{figure}[htbp]
\hspace*{-0.8cm}
 \begin{minipage}[c]{1.07\textwidth}
\centering
 \begin{minipage}[l]{0.49\textwidth}
\centering
\includegraphics[width=\textwidth]{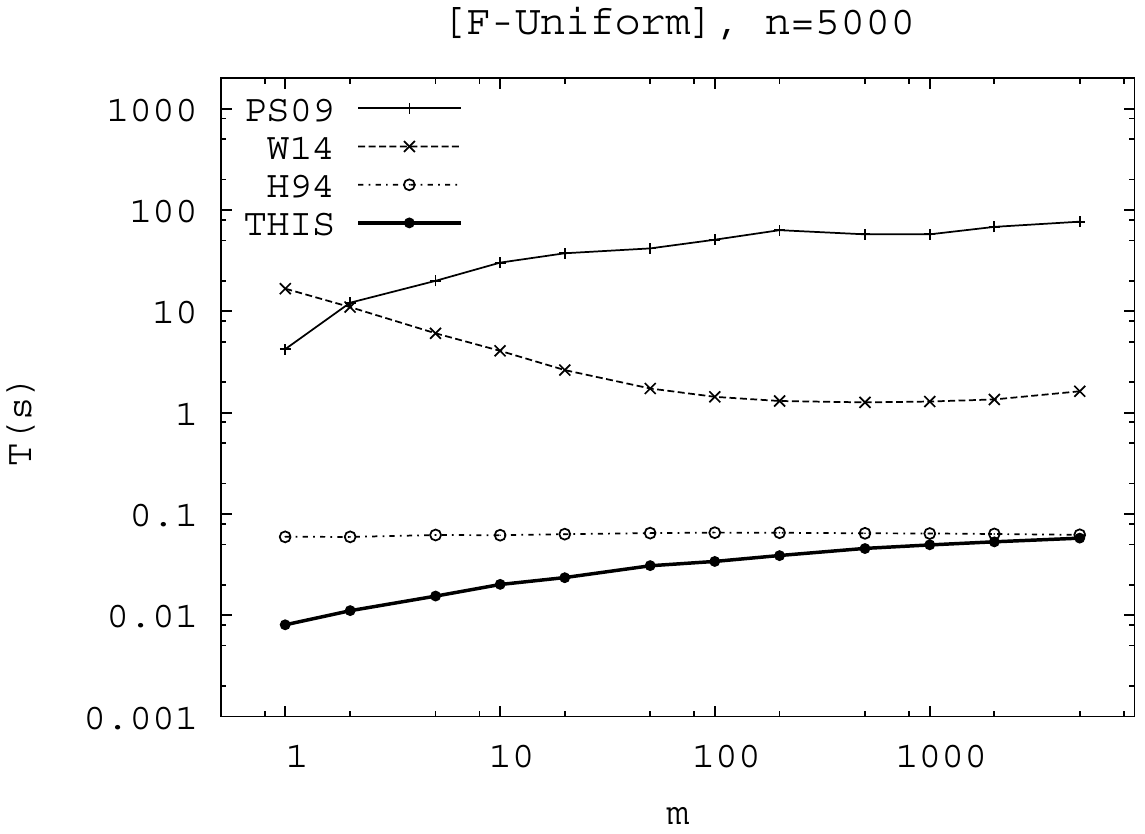}

\vspace*{0.2cm}

\includegraphics[width=\textwidth]{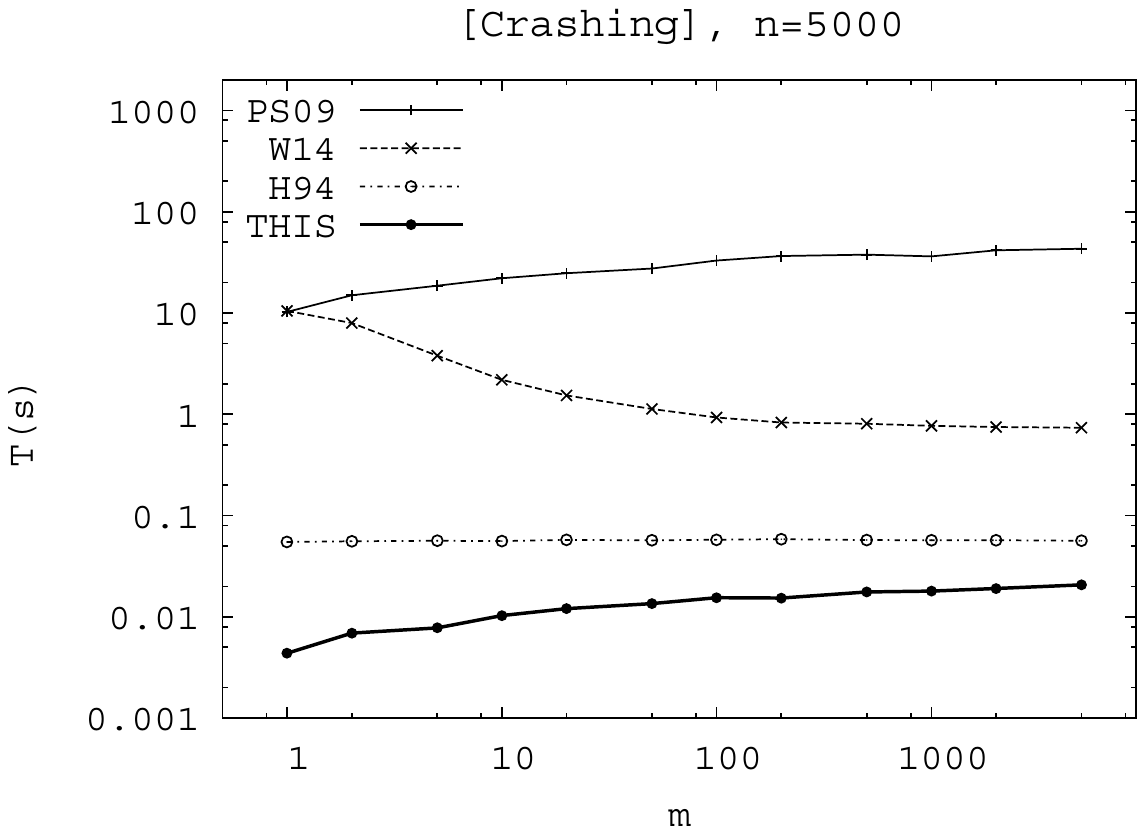}

\vspace*{0.2cm}

\includegraphics[width=\textwidth]{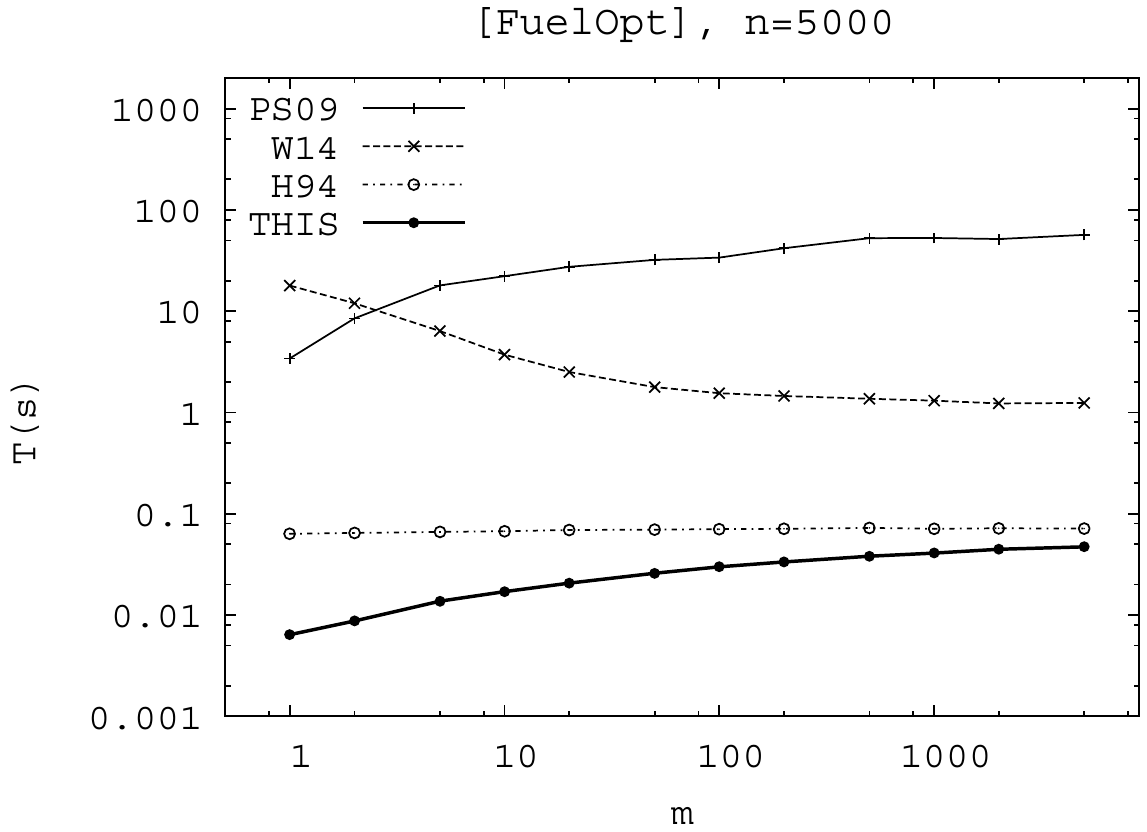}
 \end{minipage}
 \begin{minipage}[r]{0.49\textwidth}
\includegraphics[width=\textwidth]{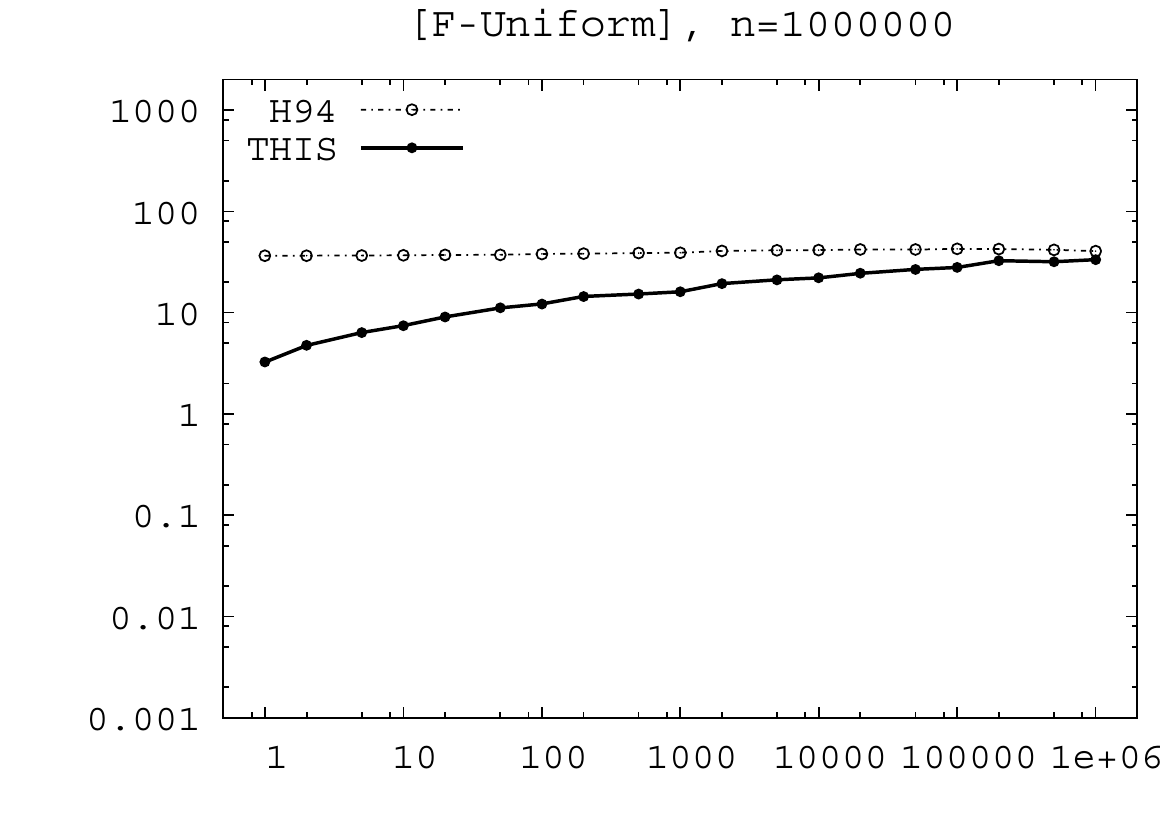}

\vspace*{0.2cm}

\includegraphics[width=\textwidth]{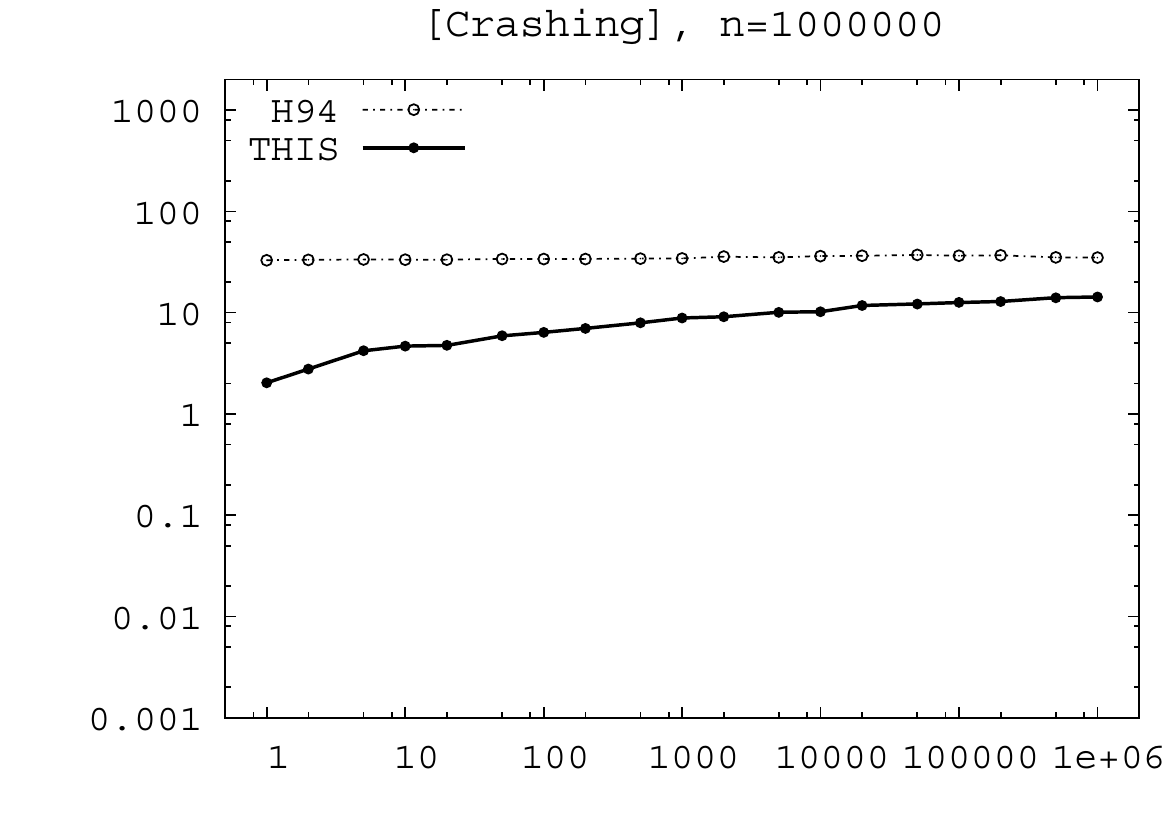}

\vspace*{0.2cm}

\includegraphics[width=\textwidth]{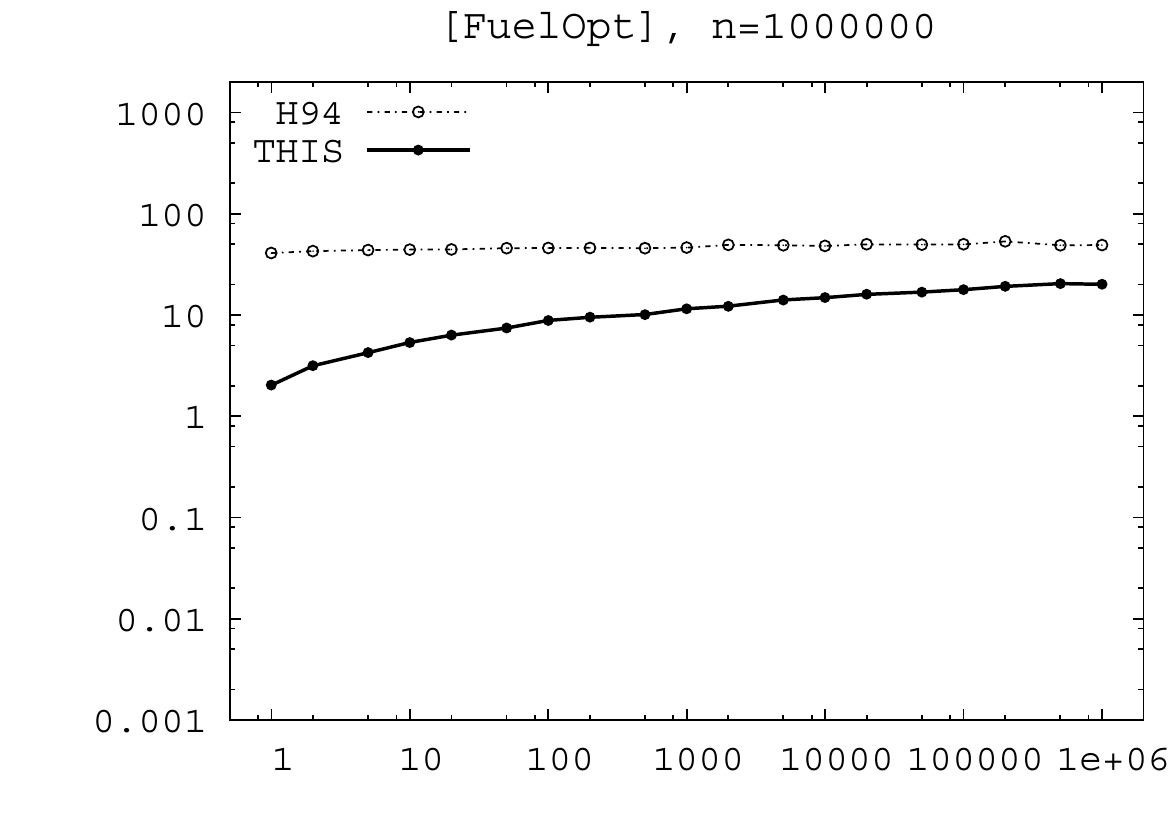}
\end{minipage} 

\vspace*{0.2cm}

\end{minipage} 
\caption{CPU Time(s) as a function of $m$. $n \in \{5000,1000000\}$. Logarithmic representation}
\label{figure-m1}
\end{figure}

First, it is remarkable that the number of active nested constraints strongly varies from one set of benchmark instances to another. One drawback of the previously-used [F] instances of \citep{Padakandla2009} is that they lead to a low number of active nested constraints, in such a way that in many cases an optimal  RAP solution obtained by relaxing all nested constraints is also the optimal NESTED solution. Some algorithms can benefit from such problem characteristics.

The five considered methods require very different CPU time to reach the optimal solution with the same precision. In all cases, the smallest time was achieved by our decomposition method.
The time taken by PS09, W14, H94 and our decomposition algorithm, as a function of $n$, is in most most cases in accordance with the theoretical complexity, cubic for PS09, quadratic for W14, and log-linear for H94 and the proposed method (Figure \ref{figure-n}). The only notable exception is problem type [F], for which the reduced number of active constraints leads to a general quadratic behavior of PS09 (instead of cubic). The CPU time of MOSEK does not exhibit a polynomial behavior on the considered problem-size range, possibly because of the preprocessing phase. 
The proposed method and H94 have a similar growth when $m=n$, but our dual-inspired decomposition algorithm performs faster in practice by a constant factor $\times$1 to $\times$10.
This can be related to the fact that our method only relies on tables, and thus hidden constants related to the use of priority-lists or union-find data structures are avoided. The bottleneck of our method (measured by means of a time profiler) is the call to the oracle for the objective function.
In H94, the call to the oracle and the management of the priority list for finding the minimum cost increment contribute equally to the largest part of the CPU time. The time taken by the Union-Find structures is not significant.

\subsection{Experiments with $\mathbf{m < n}$} 
In  a second set of experiments, the number of variables is fixed and the impact of the number of nested constraints is evaluated, with $m \in \{1,2,5,10,50,\dots,n\}$, on [F-Uniform], [Crashing] and [FuelOpt].
Two values $n=5000$ and $n=1,000,000$ were considered, to allow experiments with PS09, W14, H14 and the proposed method on medium size problems in reasonable CPU time, as well as further tests with H14 and the proposed method on large-scale instances. The CPU time as a function of $m$ is displayed in Figure~\ref{figure-m1}. 

The CPU time of H94 appears to be independent of $m$, while significant time gains can be observed for the proposed method, which is $\times$5 to $\times$20 faster than H94 on large-scale instances ($n=1,000,000$) with few nested constraints ($m=10$ or $100$). It also appears that PS09 benefits from sparser constraints. Surprisingly, sparser constraints are detrimental to W14 in practice, possibly because Equation (21) of \citep{Wang2012} is called on larger sets of variables.

\section{A note on the number of active nested constraints}
\label{active}

%
The previous experiments have shown that the number of active nested constraints in the optimal solutions tends to grow sub-linearly for the considered problems. In Table \ref{detailed-n} for example, even when $m=10^6$ the number of active nested constraints is located between $12.8$ and $88.3$ for instances with randomly generated coefficients (no ordering as in [F] or [F-Active]). To complement this observation,  we show in the following that the expected number of active nested constraints in a random optimal solution grows logarithmically with $m$ when :
\begin{enumerate}
\item $d_i = +\infty$;
\item parameters $\alpha_i$ are outcomes of i.i.d. random variables;
\item functions $f_i$ are strictly convex and differentiable;
\item and there exists a function $h$ and $\gamma_i \in \Re^{+*}$ for $i \in \{1,\dots,n\}$ satisfying $f_i(x) = \gamma_i h(x/\gamma_i)$. $\gamma_i$ are i.i.d. random variables independent from the  $\alpha_i$'s, and the vectors $(\gamma_i,\alpha_i)$ are non-colinear. \\ \vspace*{-0.1cm}
\end{enumerate}

Function shapes satisfying condition 4. are frequently encountered, e.g. in
\begin{itemize}
\item crashing: $f_i(x) =  p_i/x$ $\Rightarrow$ $h(x) = 1/x$ and  $\gamma_i = \sqrt{p_i}$;
\item fuel optimization: $f_i(x) =  p_i c_i (c_i/x)^3$ $\Rightarrow$ $h(x) = 1/x^3$ and $\gamma_i = c_i  \sqrt[4]{p_i}$;
\item any function $f_i(x) = p_i x^k$ s.t. $k \neq 1$ $\Rightarrow$  $h(x) = x^k$ and $\gamma_i = 1 / p_i^{1/(k-1)}$. \\
\end{itemize}

The first order necessary and sufficient optimality conditions of problem (\ref{RAP:1}-\ref{RAP:3}) with $x_i \in \Re^+$ for $i \in \{1,\dots,n\}$ can be written as:
\begin{small}
\begin{align}
&\mathbf{x}= (x_1,\dots,x_n) \geq \textbf{0} \text{ satisfy constraints (\ref{RAP:2}) and (\ref{RAP:3})} \label{6one} \\
&\text{for } i \in \{1,\dots,m\} \text{ and } j \in \{s[i-1]+1,\dots,s[i]-1\}, f'_j(x_j) = f'_{j+1}(x_{j+1}) \label{6two}  \\
&\text{for } i \in \{1,\dots,m-1\} \text{ and } j=s[i], \
\begin{cases}
\textbf{either } f'_{j}(x_{j}) =f'_{j+1}(x_{j+1}) \\
\textbf{or } f'_{j}(x_{j}) < f'_{j+1}(x_{j+1})  \textbf{ and }  \sum_{k=1}^{j} x_k = a_i
\end{cases}  \label{6three} 
\end{align}
\end{small}

If $f_i(x) = \gamma_i h(\frac{x}{\gamma_i})$, then $f'_i(x) = h'(\frac{x}{\gamma_i})$, and with the strict convexity the necessary and sufficient conditions (\ref{6two}) and (\ref{6three}) become:
\begin{small}
\begin{align}
&\text{for } i \in \{1,\dots,m\} \text{ and } j \in \{s[i-1]+1,\dots,s[i]-1\},  \frac{x_{s[i]}}{\gamma_{s[i]}} =   \frac{x_{s[i]+1}}{\gamma_{s[i]+1}} \tag{\ref{6two}b}   \label{6twobis}\\
&\text{for } i \in \{1,\dots,m-1\}  \text{ and } j=s[i], \
\begin{cases}
\textbf{either }  \frac{x_{j}}{\gamma_{j}} =   \frac{x_{j+1}}{\gamma_{j+1}}  \\
\textbf{or } \frac{x_{j}}{\gamma_{j}} <   \frac{x_{j+1}}{\gamma_{j+1}}  \textbf{ and }  \sum_{k=1}^{j} x_k = a_i \tag{\ref{6three}b} \label{6threebis}
\end{cases}
\end{align}
\end{small}

\begin{figure}[htbp]
\centering
\includegraphics[width=0.8\textwidth]{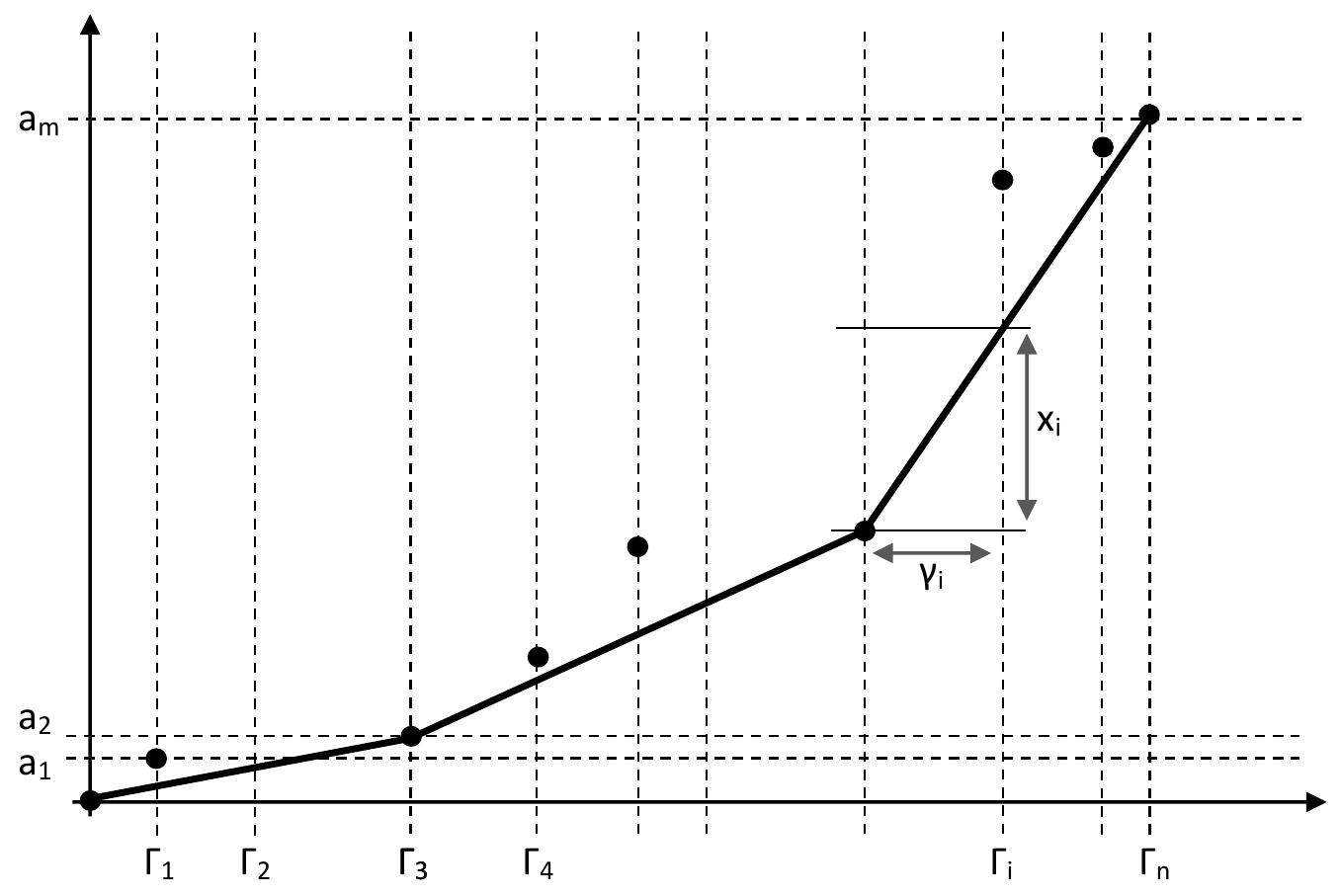}
\caption{Reduction of NESTED to a convex hull computation.
Example with $n=10$, $m=8$ and $s = (1,3,4,5,7,8,9)$.}
\label{graph-method}
\end{figure}

Define $\Gamma_i = \sum_{k=1}^{i} \gamma_k$ for $i \in \{0,\dots,n\}$.
As illustrated on Figure \ref{graph-method}, searching for a solution satisfying  (\ref{RAP:2}), (\ref{RAP:3}),  (\ref{6twobis}) and  (\ref{6threebis}) reduces to computing the convex hull of the set of points $\mathcal{P}$ such that
\begin{align}
 \mathcal{P} = \{  (\Gamma_{s[j]} , a_j) \ | \ j \in \{0,\dots,m\} \}.
\end{align}
Let $\Phi : [0,\Gamma_{n}] \rightarrow [0,B]$ be the curve associated with the lower part of the convex hull, in boldface on Figure \ref{graph-method}. Then, the solution defined as $x_i = \Phi(\Gamma_{i}) - \Phi(\Gamma_{i-1})$ for $i \in \{1,\dots,n\}$ satisfies all previously-mentioned conditions since
\begin{itemize}
\item $\Phi$ is below the points $p_j$, hence satisfying (\ref{RAP:2}); 
\item $p_m$ is part of the convex hull, thus satisfying (\ref{RAP:3});
\item $\Phi(z) \geq 0$ for $z \in [0,\Gamma_{n}]$ since all $p_j$ coordinates are non-negative, hence $\mathbf{x} \geq \textbf{0}$;
\item the slope of $\Phi$ is constant between vertices of the convex hull (\ref{6twobis});
\item and the slope of $\Phi$ only increases when meeting a vertex (\ref{6threebis}). \\
\end{itemize}

The expected number of vertices of a convex hull with random points is at the core of an extensive literature. We refer to \citep{Deltheil1920} for early studies, and \citep{Majumdar2010} for a recent review. 
Consider a randomly-generated NESTED problem, such that $\gamma_j$ for $j \in \{1,\dots,m\}$  and $\alpha_j$ for $j \in \{1,\dots,n\}$ are i.i.d. random variables. If the distribution is such that all vectors $(\gamma_j,\alpha_j)$ for $j \in \{1,\dots,m\}$ are non co-linear, then the expected number of points on the convex hull grows as $O(\log m)$ \citep{Baxter1961}. Equivalently, there are $O(\log m)$ expected active nested constraints in the solution.

Note that a generalization of the previous reasoning is necessary to fully explain the results of our experiments since we considered $d_i \neq \infty$. 
Assuming that the same result holds in this more general case, then the amortized complexity of some methods such as \citep{Padakandla2009} on randomly generated instances may be significantly better than the worst case. Indeed, this method iterates on the number of active constraints in an outer loop. The number of active constraints has no impact on the complexity and CPU time of the proposed method, but further pruning techniques may be investigated to eliminate constraints on the fly. Finally, the graphical approach used in this analysis leads to a strongly polynomial algorithm in $O(n + m \log m)$ for an interesting class of problems, and is worth further investigation on its own.

\section{Conclusions}
\label{concl}
A dual-inspired approach has been introduced for NESTED resource allocation problems. The method solves NESTED as a hierarchy of simple resource allocation problems. The best known complexity of $O(n \log n \log \frac{B}{n} )$ is attained for problems with as many nested constraints as variables, and a new best-known complexity of $O(n \log m \log \frac{B}{n} )$ is achieved for problems with $n$ variables and $\log m = o(\log n)$ nested constraints. Extensive computational experiments highlight significant CPU time gains in comparison to other state-of-the-art methods on a wide range of problem instances with up to one million tasks.

The proposed algorithm relies on different principles than the previous state-of-the-art scaled greedy method. As such, it is not bound to the same methodological limitations and may be generalized to some problem settings with non-polymatroidal constraints, e.g., allocation problems with nested upper and lower constraints, which are also related to various key applications. Further pruning techniques exploiting the reduced number of active nested constraints can be designed and the geometric approach of Section \ref{active} can be further investigated, aiming for generalization and an increased understanding of its scope of application.
Finally, promising research perspectives relate to the extension of these techniques for various application fields, such as telecommunications and image processing, which can require to solve huge problems with similar formulations.

\bibliographystyle{plain}

\end{document}